\documentclass[a4paper,onecolumn,11pt,accepted=2024-03-23]{quantumarticle}
\pdfoutput=1
\usepackage[english]{babel}
\usepackage{hyperref}
\usepackage{tikz}
\usepackage{lipsum}
\usepackage{amsmath,amsfonts,amssymb,amsthm,amscd}
\usepackage{graphicx}
\usepackage{revsymb}
\usepackage{color,comment}
\usepackage[T1]{fontenc}
\usepackage[latin1]{inputenc}

\newtheorem{theorem}{Theorem}[section]
\newtheorem{lemma}[theorem]{Lemma}

\newtheorem{proposition}[theorem]{Proposition}
\newtheorem{remark}[theorem]{Remark}
\newtheorem{defi/prop}[theorem]{Definition/Proposition}

\newcommand{\N}{\mathbf{N}}

\newcommand{\C}{\mathbf{C}}
\renewcommand{\P}{\mathbf{P}}
\newcommand{\rE}{\mathrm{E}}
\newcommand{\A}{\mathrm{A}}
\newcommand{\B}{\mathrm{B}}
\renewcommand{\H}{\mathrm{H}}
\newcommand{\cL}{\mathcal{L}}

\renewcommand{\leq}{\leqslant}
\renewcommand{\geq}{\geqslant}

\newcommand{\st}{\  : \ }

\newcommand{\Sym}{\mathrm{Sym}}
\newcommand{\Id}{\openone}

\newcommand{\cU}{\mathcal{U}}
\newcommand{\cD}{\mathcal{D}}
\newcommand{\cN}{\mathcal{N}}
\newcommand{\cM}{\mathcal{M}}

\DeclareMathOperator{\tr}{Tr}
\DeclareMathOperator{\E}{\mathbf{E}}

\newcommand{\braket}[2]{\langle #1 | #2\rangle}
\newcommand{\ketbra}[2]{| #1 \rangle\!\langle #2 |}
\newcommand{\bra}[1]{\langle #1 |}
\newcommand{\ket}[1]{| #1 \rangle}
\newcommand{\proj}[1]{| #1 \rangle\!\langle #1 |}

\begin{document}

\title{Approximating quantum channels by completely positive maps with small Kraus rank}

\author{C\'{e}cilia Lancien}
\affiliation{Institut Fourier \& CNRS, Universit\'{e} Grenoble Alpes, 38610 Gi\`{e}res, France}
\email{cecilia.lancien@univ-grenoble-alpes.fr}
\author{Andreas Winter}
\email{andreas.winter@uab.cat}
\affiliation{Departament de F\'{\i}sica: Grup d'Informaci\'{o} Qu\`{a}ntica, Universitat Aut\`{o}noma de Barcelona, 08193 Bellaterra (Barcelona), Spain \& Instituci\'{o} Catalana de Recerca i Estudis Avan\c{c}ats, 08010 Barcelona, Spain}

\maketitle

\begin{abstract}
We study the problem of approximating a quantum channel by one with as few Kraus operators as possible (in the sense that, for any input state, the output states of the two channels should be close to one another). Our main result is that any quantum channel mapping states on some input Hilbert space $\A$ to states on some output Hilbert space $\B$ can be compressed into one with order $d\log d$ Kraus operators, where $d=\max(|\A|,|\B|)$, hence much less than $|\A||\B|$. In the case where the channel's outputs are all very mixed, this can be improved to order $d$. We discuss the optimality of this result as well as some consequences.
\end{abstract}

\section{Introduction}

Quantum channels are the most general framework in which the transformations that a quantum system may undergo are described. These are defined as completely positive and trace preserving (CPTP) maps from the set of bounded operators on some input Hilbert space $\A$ to the set of bounded operators on some output Hilbert space $\B$. Indeed, to be a physically valid evolution in the open system setting, a linear map $\cN$ has to preserve quantum states (i.e.~positive semi-definiteness and unit-trace conditions) even when tensorized with the identity map $\mathcal{I}$ on an auxiliary system.

Let us fix here once and for all some notation that we will use repeatedly in the remainder of the paper: Given a Hilbert space $\mathrm{H}$, we shall denote by $\cL(\mathrm{H})$ the set of linear operators on $\mathrm{H}$, and by $\cD(\mathrm{H})$ the set of density operators (i.e.~positive semi-definite and trace $1$ operators) on $\mathrm{H}$. Also, whenever $\mathrm{H}$ is finite dimensional (which will be the case of all the Hilbert spaces we will deal with in the sequel) we shall denote by $|\mathrm{H}|$ its dimension.

So assume from now on that the Hilbert spaces $\A$ and $\B$ are finite dimensional. Then, we know by Choi's representation theorem \cite{Choi} that a CPTP map $\mathcal{N}:\cL(\A)\rightarrow\cL(\B)$ can always be written as
\begin{equation} \label{eq:Kraus} \mathcal{N}: X\in\cL(\A) \mapsto \sum_{i=1}^s K_i X K_i^{\dagger} \in\cL(\B), \end{equation}
where the operators $K_i:\A\rightarrow\B$, $1\leq i\leq s$, are called the Kraus operators of $\mathcal{N}$ and satisfy the normalization relation $\sum_{i=1}^s K_i^{\dagger}K_i = \Id_{\A}$. The minimal $s\in\N$ such that $\mathcal{N}$ can be decomposed in the Kraus form \eqref{eq:Kraus} is called the Kraus rank of $\mathcal{N}$, which we shall denote by $r_K(\mathcal{N})$. By Stinespring's dilatation theorem \cite{Stine}, another alternative way of characterizing a CPTP map $\mathcal{N}:\cL(\A)\rightarrow\cL(\B)$ is as follows
\begin{equation} \label{eq:Stinespring} \mathcal{N}: X\in\cL(\A) \mapsto \tr_{\rE}\left(VXV^{\dagger}\right) \in\cL(\B), \end{equation}
for some environment Hilbert space $\rE$ and some isometry $V:\A\hookrightarrow\B\otimes\rE$ (i.e.~$V^{\dagger}V=\Id_{\A}$). In such picture, $r_K(\mathcal{N})$ is then nothing else than the minimal environment dimension $|\rE|\in\N$ such that $\mathcal{N}$ may be expressed in the Stinespring form \eqref{eq:Stinespring}. It may be worth pointing out that there is a lot of freedom in representation \eqref{eq:Kraus}: two sets of Kraus operators $\{K_i,\ 1\leq i\leq s\}$ and $\{L_i,\ 1\leq i\leq s\}$ give rise to the same quantum channel as soon as there exists a unitary $U$ on $\C^s$ such that, for all $1\leq i\leq s$, $L_i=\sum_{j=1}^sU_{ij}K_j$. On the contrary, representation \eqref{eq:Stinespring} is essentially unique, up to the (usually irrelevant) transformation $V\mapsto(\Id\otimes U) V$, for $U$ a unitary on $\C^s$. That is why we will often prefer working with the latter than with the former.

Yet another way of viewing the Kraus rank of a CPTP map $\mathcal{N}:\cL(\A)\rightarrow\cL(\B)$ is as the rank of its associated Choi-Jamio\l{}kowski state. Denoting by $\psi$ the maximally entangled state on $\A\otimes\A$, the latter is defined as the state $\tau(\cN) = \mathcal{I}\otimes\cN(\psi)$ on $\A\otimes\B$.
Consequently, it holds that any quantum channel from $\A$ to $\B$ has Kraus rank at most $|\A||\B|$. And the extremal such quantum channels are those with Kraus rank less than $|\A|$. In particular, the case $r_K(\cN)=1$ corresponds to $\cN$ being a unitary, hence reversible, evolution, whereas whenever $r_K(\cN)>1$, one can view $\cN$ as a noisy summary of a unitary evolution on a larger system. The Kraus rank of a quantum channel can thus legitimately be seen as a measure of its ``complexity'': it quantifies the minimal amount of ancillary resources needed to implement it (or equivalently the amount of degrees of freedom in it that one is ignorant of). A natural question in this context would therefore be: given any quantum channel, is it possible to reduce its complexity while not affecting too much its action, or in other words to find a channel with much smaller Kraus rank which approximately simulates it?

One last definition we shall need concerning CP maps is the following: the conjugate (or dual) of a CP map $\cN:\cL(\A)\rightarrow\cL(\B)$ is the CP map $\cN^*:\cL(\B)\rightarrow\cL(\A)$ defined by
\[ \forall\ X\in\cL(\A),\ \forall\ Y\in\cL(\B),\ \tr(\cN(X)Y) = \tr(X\cN^*(Y)). \]
It is characterized as well by saying that $\{K_i,\ 1\leq i\leq s\}$ is a set of Kraus operators for $\cN$ if and only if $\{L_i=K_i^{\dagger},\ 1\leq i\leq s\}$ is a set of Kraus operators for $\cN^*$. Hence obviously, $\cN$ and $\cN^*$ have same Kraus rank, while the trace-preservingness condition $\sum_{i=1}^sK_i^{\dagger}K_i=\Id$ for $\cN$ is equivalent to the unitality condition $\sum_{i=1}^sL_iL_i^{\dagger}=\Id$ for $\cN^*$.

The remainder of this paper is organized as follows. In Section \ref{sec:background} we gather all needed background on quantum channel approximation that we are interested in. This includes precise definitions, previous works in this direction, etc. Our main results are then stated and commented in Section \ref{sec:results}, while their proofs are relegated to Section \ref{sec:proofs}. In Section \ref{sec:applications} we present several corollaries, which have applications in quantum data hiding and locking, amongst other. We finally discuss some open questions in Section \ref{sec:discussion}.

\section{Quantum channel approximation: definitions and already known facts}
\label{sec:background}

Before going any further, we need to specify what we mean by ``approximating a quantum channel'', since indeed, several definitions of approximation may be considered. In our setting, the most natural one is probably that of approximation in $(1{\rightarrow}1)$-norm: given CPTP maps $\mathcal{N},\widehat{\mathcal{N}}:\cL(\A)\rightarrow\cL(\B)$, we will say that $\widehat{\cN}$ is an $\varepsilon$-approximation of $\cN$ in $(1{\rightarrow}1)$-norm, where $\varepsilon>0$ is some fixed parameter, if
\begin{equation} \label{eq:def_1->1} \forall\ \varrho\in\cD(\A),\ \left\|\widehat{\cN}(\varrho)-\cN(\varrho)\right\|_1 \leq \varepsilon. \end{equation}
At first sight it might appear that an even more natural error quantification in such a context would be in terms of the completely-bounded $(1{\rightarrow}1)$-norm (aka diamond norm) \cite{diamond}. That is, in order to call $\widehat{\cN}$ an $\varepsilon$-approximation of $\cN$, we would require that, for any Hilbert space $\A'$,
\begin{equation} \label{eq:def_cb-1->1} \forall\ \varrho\in\cD(\A\otimes\A'),\ \left\|\widehat{\cN}\otimes\mathcal{I}(\varrho)-\cN\otimes\mathcal{I}(\varrho)\right\|_1 \leq \varepsilon. \end{equation}
Nevertheless, this notion of approximation is too strong for our purposes. Indeed, if $\mathcal{N}$ and $\widehat{\mathcal{N}}$ satisfy Equation \eqref{eq:def_cb-1->1}, it implies in particular that their associated Choi-Jamio\l{}kowski states have to be $\varepsilon$-close in trace-norm. And this, in general, is possible only if $\mathcal{N}$ and $\widehat{\mathcal{N}}$ have the same, or at least comparable, number of Kraus operators, so that no environment dimensionality reduction can be achieved.

The question of quantum channel compression has already been studied in one specific case, which is the one of the fully randomizing (or depolarizing) channel. Let us recall what is known there. The fully randomizing channel $\mathcal{R}:\cL(\A)\rightarrow\cL(\A)$ is the CPTP map with same input and output spaces defined by
\[ \mathcal{R}:X\in\cL(\A)\mapsto (\tr X)\frac{\Id}{|\A|} \in\cL(\A), \]
so that, in particular, all input states $\varrho\in\cD(\A)$ are sent to the maximally mixed state $\Id/|\A|\in\cD(\A)$. $\mathcal{R}$ has maximal Kraus rank $|\A|^2$ (because $\tau(\mathcal{R})$ is simply $\Id/|\A|^2$, and hence has rank $|\A|^2$). This was of course to be expected, if adhering to the intuitive idea that the bigger is the Kraus rank of channel, the noisier is the channel. One possible minimal Kraus decomposition for $\mathcal{R}$ is
\[ \mathcal{R}:X\in\cL(\A)\mapsto  \frac{1}{|\A|^2}\sum_{j,k=1}^{|\A|} V_{jk}X V_{jk}^{\dagger} \in\cL(\A), \]
where for each $1\leq j,k\leq |\A|$, $V_{jk}=X^j Z^k$ with $X$ and $Z$ the generalized 
Pauli shift and phase operators on $\A$.
It was initially established in \cite{HLSW} and later improved in \cite{Aubrun} that there exist almost randomizing channels with drastically smaller Kraus rank. More specifically, the following was proved: for any $0<\varepsilon<1$, the CPTP map $\mathcal{R}$ can be $\varepsilon$-approximated in $(1{\rightarrow}1)$-norm by a CPTP map $\widehat{\mathcal{R}}$ with Kraus rank at most $C|\A|/\varepsilon^2$, where $C>0$ is a constant that is independent of $|\A|$ and $\varepsilon$. Actually, something stronger was established, namely 
\[ \forall\ \varrho\in\cD(\A),\ \left\|\widehat{\mathcal{R}}(\varrho)-\mathcal{R}(\varrho)\right\|_{\infty} \leq \frac{\varepsilon}{|\A|}, \]
which obviously implies that, for any $1\leq p\leq\infty$, $\widehat{\mathcal{R}}$ is an $\varepsilon$-approximation of $\mathcal{R}$ in $(1{\rightarrow}p)$-norm, in the sense that
\begin{equation} \label{eq:def_1->p} \forall\ \varrho\in\cD(\A),\ \left\|\widehat{\mathcal{R}}(\varrho)-\mathcal{R}(\varrho)\right\|_p \leq \frac{\varepsilon}{|\A|^{1-1/p}}. \end{equation}

The question we investigate here is whether such kind of statement actually holds true for any channel. Note however that, for a channel which is not the fully randomizing one, imposing an approximation in Schatten-$p$-norm up to an error $\varepsilon/|\B|^{1-1/p}$, as the error appearing in Equation \eqref{eq:def_1->p}, is potentially too strong. Indeed, the fully randomizing channel is such that all outputs are the maximally mixed state, and thus have $p$-norm equal to $1/|\B|^{1-1/p}$. But other channels might have outputs which are much less mixed, and thus have much larger $p$-norm. It would therefore seem more accurate to quantify closeness in terms of relative error. Hence, given a CPTP map $\mathcal{N}:\cL(\A)\rightarrow\cL(\B)$, we would be interested in finding a CPTP map $\widehat{\mathcal{N}}:\cL(\A)\rightarrow\cL(\B)$ with Kraus rank as small as possible, and such that
\begin{equation} \label{eq:approx-p-relative} \forall\ \varrho\in\cD(\A),\ \left\|\widehat{\cN}(\varrho)-\cN(\varrho)\right\|_p \leq \varepsilon\left\|\cN(\varrho)\right\|_p. \end{equation}

\section{Statement of the main results}
\label{sec:results}

Our strategy in order to prove an approximation result of the form given by Equation \eqref{eq:def_1->1} will be to first establish approximation in terms of operator ordering of the outputs for all inputs. Concretely, given a CPTP map $\mathcal{N}:\cL(\A)\rightarrow\cL(\B)$, the idea will be to look for a CPTP map $\widehat{\mathcal{N}}:\cL(\A)\rightarrow\cL(\B)$ with Kraus rank as small as possible, and such that
\begin{equation} \label{eq:approxmain-ideal} \forall\ \varrho\in\cD(\A),\ (1-\varepsilon)\cN(\varrho) \leq \widehat{\cN}(\varrho)\leq (1+\varepsilon)\cN(\varrho). \end{equation}
The approximation statement that we establish in Theorem \ref{th:main} is not exactly that of Equation \eqref{eq:approxmain-ideal}, but very close to it.

\begin{theorem} \label{th:main}
Fix $0<\varepsilon<1$ and let $\cN:\cL(\A)\rightarrow\cL(\B)$ be a CPTP map with Kraus rank $|\rE|\geq |\A|,|\B|$. Then, there exists a CP map $\widehat{\cN}:\cL(\A)\rightarrow\cL(\B)$ with Kraus rank at most $C\max(|\A|,|\B|)\log(|\rE|/\varepsilon)/\varepsilon^2$ (where $C>0$ is a universal constant) and such that
\begin{equation} \label{eq:approxmain} \forall\ \varrho\in\cD(\A),\ -\varepsilon\left(\cN(\varrho) -\frac{\Id}{|\B|}\right) \leq \widehat{\cN}(\varrho)-\cN(\varrho)\leq \varepsilon\left(\cN(\varrho) +\frac{\Id}{|\B|}\right). \end{equation}
\end{theorem}

\begin{remark} \label{remark:schatten}
Note that if $\widehat{\cN}$ satisfies Equation \eqref{eq:approxmain}, then it especially implies that it approximates $\cN$ in any Schatten-norm in a sense close to that of Equation \eqref{eq:approx-p-relative}, namely
\[ \forall\ p\in\N,\ \forall\ \varrho\in\cD(\A),\ \left\|\widehat{\cN}(\varrho)-\cN(\varrho)\right\|_p \leq \varepsilon\left(\left\|\cN(\varrho)\right\|_p+\frac{1}{|\B|^{1-1/p}}\right). \]
In particular, we have the $(1{\rightarrow}1)$-norm approximation of $\cN$ by $\widehat{\cN}$
\[ \forall\ \varrho\in\cD(\A),\ \left\|\widehat{\cN}(\varrho)-\cN(\varrho)\right\|_1 \leq 2\varepsilon,  \]
in which we can further impose that $\widehat{\cN}$ is strictly, and not just up to an error $2\varepsilon$, trace preserving (see the proof of Theorem \ref{th:main}).
\end{remark}

The main interest of the approximation statements in Theorem \ref{th:main} and Remark \ref{remark:schatten} is their universality. They indeed establish that any quantum channel can be approximated by one that has Kraus rank of order $\max(|\A|,|\B|)\log(|\rE|)$. They of course do not exclude the existence of channels for which a better compression would be possible, but provide an achievable compression in the worst-case scenario.

One important question at this point is the one of optimality in the above results. A first obvious observation to make in order to answer it is the following: if a CP map has Kraus rank $s$, then it necessarily sends rank $1$ inputs to output states of rank at most $s$. This is of course informative only if $s$ is smaller than the output space dimension. But as we shall see, having this in mind will be useful to prove that certain channels cannot be compressed further than as guaranteed by Theorem \ref{th:main}.

Our constructions will be based on the existence of so-called tight normalized frames. Namely, for any $N,d\in\N$ with $N\geq d$, there exist unit vectors $\ket{\psi_1},\ldots,\ket{\psi_N}$ in $\C^d$ such that
\[ \frac{1}{N}\sum_{k=1}^N\ketbra{\psi_k}{\psi_k}=\frac{\Id}{d}. \]
Denoting by $\{\ket{j},\ 1\leq j\leq d\}$ an orthonormal basis of $\C^d$, a possible way of constructing such vectors is e.g.~to make the choice
\begin{equation} \label{eq:frame'} \forall\ 1\leq k\leq N,\ \ket{\psi_k}=\frac{1}{\sqrt{d}}\sum_{j=1}^d e^{2i\pi jk/N}\ket{j}. \end{equation}
Note that if this so, then any basis vector $\ket{j}$, $1\leq j\leq d$, is such that, for each $1\leq k\leq N$, $|\braket{\psi_k}{j}|^2=1/d$.

Let us now come back to our objective. What we want to exhibit here are CPTP maps $\cN:\cL(\A)\rightarrow\cL(\B)$ with either one or the other of the following two properties: if a CP map $\widehat{\cN}:\cL(\A)\rightarrow\cL(\B)$ satisfies
\begin{equation} \label{eq:approx1} \forall\ R\in\mathcal{H}_+(\A),\
(1-\varepsilon)\cN(R)-\varepsilon(\tr R)\frac{\Id}{|\B|} \leq \widehat{\cN}(R) \leq (1+\varepsilon)\cN(R)+\varepsilon(\tr R)\frac{\Id}{|\B|}, \end{equation}
then it necessarily has to be such that either $r_K(\widehat{\cN})\geq |\A|$ or $r_K(\widehat{\cN})\geq |\B|$. Besides, note that the CP maps $\cN,\widehat{\cN}$ fulfilling condition \eqref{eq:approx1} above is equivalent to the conjugate CP maps $\cN^*,\widehat{\cN}^*$ fulfilling condition \eqref{eq:approx2} below
\begin{equation} \label{eq:approx2} \forall\ R\in\mathcal{H}_+(\B),\
(1-\varepsilon)\cN^*(R)-\varepsilon(\tr R)\frac{\Id}{|\B|} \leq \widehat{\cN}^*(R) \leq (1+\varepsilon)\cN^*(R)+\varepsilon(\tr R)\frac{\Id}{|\B|}. \end{equation}
Indeed, we have the following chain of equivalences
\begin{align*}
& \forall\ R\in\mathcal{H}_+(\A),\ \widehat{\cN}(R) \leq (1+\varepsilon)\cN(R)+\varepsilon(\tr R)\frac{\Id}{|\B|} \\ 
\Leftrightarrow\ & \forall\ R\in\mathcal{H}_+(\A),S\in\mathcal{H}_+(\B),\ \tr\left(\widehat{\cN}(R)S\right) \leq (1+\varepsilon)\tr\left(\cN(R)S\right)+\varepsilon(\tr R)\tr\left(\frac{\Id}{|\B|}S\right) \\
\Leftrightarrow\ & \forall\ R\in\mathcal{H}_+(\A),S\in\mathcal{H}_+(\B),\ \tr\left(R\widehat{\cN}^*(S)\right) \leq (1+\varepsilon)\tr\left(R\cN^*(S)\right)+\varepsilon(\tr S)\tr\left(\frac{\Id}{|\B|}R\right) \\
\Leftrightarrow\ & \forall\ S\in\mathcal{H}_+(\B),\ \widehat{\cN}^*(S) \leq (1+\varepsilon)\cN^*(S)+\varepsilon(\tr S)\frac{\Id}{|\B|},
\end{align*}
where the first and third equivalences are by the characterization of positive semidefinite matrices, and the second equivalence is by definition of conjugate maps. And similarly for the other inequality.
Depending on what we want to establish, it will be more convenient to work with either requirement \eqref{eq:approx1} or requirement \eqref{eq:approx2} 


Assume first of all that $|\B|\geq |\A|$, and consider $\cM:\cL(\A)\rightarrow\cL(\B)$ a so-called quantum-classical channel (aka measurement). More specifically, define the CPTP map
\begin{equation} \label{eq:qc} \mathcal{M}:X\in\cL(\A)\mapsto\frac{|\A|}{|\B|}\sum_{i=1}^{|\B|}\bra{\psi_i}X\ket{\psi_i}\ketbra{x_i}{x_i} \in\cL(\B), \end{equation}
where $\{\ket{x_i},\ 1\leq i\leq|\B|\}$ is an orthonormal basis of $\B$ and $\ket{\psi_1},\ldots,\ket{\psi_{|\B|}}$ are unit vectors of $\A$, defined in terms of an orthonormal basis $\{\ket{j},\ 1\leq j\leq|\A|\}$ of $\A$ as by Equation \eqref{eq:frame'}. Note that this tight normalized frame assumption implies that $\left\{(|\A|/|\B|)\ketbra{\psi_i}{\psi_i}\right\}_{1\leq i\leq|\B|}$ forms a rank-$1$ POVM on $\A$ (hence a posteriori the justification of the denomination for $\cM$). Setting, for each $1\leq i\leq|\B|$, $K_i=\sqrt{|\A|/|\B|}\,\ketbra{x_i}{\psi_i}$, we can clearly re-write $\mathcal{M}:X\in\cL(\A)\mapsto \sum_{i=1}^{|\B|}K_iXK_i^{\dagger} \in\cL(\B)$, so $r_K(\cM)\leq |\B|$. And what we actually want to show is that it is even impossible to approximate $\cM$ in the sense of Theorem \ref{th:main} with strictly less than $|\B|$ Kraus operators. Observe that by construction, say, $\ket{1}$ is such that, for each $1\leq i\leq |\B|$, $|\braket{\psi_i}{1}|^2=1/|\A|$, so that $\cM(\ketbra{1}{1})= \Id/|\B|$. Yet, assume that $\widehat{\cM}:\cL(\A)\rightarrow\cL(\B)$ is a CPTP map such that $\cM,\widehat{\cM}$ fulfill Equation \eqref{eq:approx1} for some $0<\varepsilon<1/2$. Then, the l.h.s.~of Equation \eqref{eq:approx1} yields in particular, $\widehat{\cM}(\ketbra{1}{1})\geq (1-2\varepsilon)\,\Id/|\B|$, so that $\widehat{\cM}(\ketbra{1}{1})$ has to have full rank. And therefore, it cannot be that $r_K(\widehat{\mathcal{M}})<|\B|$.

Assume now that $|\A|\geq |\B|$, and consider $\cN:\cL(\A)\rightarrow\cL(\B)$ a so-called classical-quantum channel. More specifically, define the CPTP map
\begin{equation} \label{eq:cq} \mathcal{N}:X\in\cL(\A)\mapsto\sum_{i=1}^{|\A|}\bra{x_i}X\ket{x_i}\ketbra{\psi_i}{\psi_i} \in\cL(\B), \end{equation}
where $\{\ket{x_i},\ 1\leq i\leq|\A|\}$ is an orthonormal basis of $\A$ and $\ket{\psi_1},\ldots,\ket{\psi_{|\A|}}$ are unit vectors in $\B$. Setting, for each $1\leq i\leq|\A|$, $K_i=\ketbra{\psi_i}{x_i}$, we can clearly re-write $\mathcal{N}:X\in\cL(\A)\mapsto \sum_{i=1}^{|\A|}K_iXK_i^{\dagger} \in\cL(\B)$, so $r_K(\cN)\leq |\A|$. Now, we want to show that, at least for certain choices of $\ket{\psi_1},\ldots,\ket{\psi_{|\A|}}$, it is even impossible to approximate $\cN$ in the sense of Theorem \ref{th:main} with strictly less than $|\A|$ Kraus operators. For that, we impose that they are defined in terms of an orthonormal basis $\{\ket{j},\ 1\leq j\leq|\B|\}$ of $\B$ as by Equation \eqref{eq:frame'}. Since the conjugate of $\cN$ is the CP unital map
\[ \mathcal{N}^*:X\in\cL(\B)\mapsto\sum_{i=1}^{|\A|}\bra{\psi_i}X\ket{\psi_i}\ketbra{x_i}{x_i} \in\cL(\A), \]
we have in this case that $\mathcal{M}=(|\B|/|\A|)\cN^*$ is precisely of the form \eqref{eq:qc} (with the roles of $\A$ and $\B$ switched). Hence, as we already showed, if $\widehat{\cM}:\cL(\B)\rightarrow\cL(\A)$ is a CPTP map such that $\cM,\widehat{\cM}$ fulfil Equation \eqref{eq:approx1} (with the roles of $\A$ and $\B$ switched) for some $0<\varepsilon<1/2$, then it cannot be that $r_K(\widehat{\mathcal{M}})<|\A|$. This means equivalently that if $\widehat{\cN}^*:\cL(\B)\rightarrow\cL(\A)$ is a CP map such that $\cN^*,\widehat{\cN}^*$ fulfil Equation \eqref{eq:approx2} for some $0<\varepsilon<1/2$, then it cannot be that $r_K(\widehat{\mathcal{N}}) = r_K(\widehat{\mathcal{N}}^*)<|\A|$.


Summarizing, we just established that $n\geq \max (|\A|,|\B|)$ is for sure necessary in Theorem \ref{th:main}. But it is not clear whether or not the $\log |\rE|$ factor can be removed. In the case of ``well-behaved'' channels, whose range is only composed of sufficiently mixed states, we can answer affirmatively, as an immediate implication of Theorem \ref{th:main'} below. However, we leave the question open in general.

\begin{theorem} \label{th:main'}
Fix $0<\varepsilon<1$ and let $\cN:\cL(\A)\rightarrow\cL(\B)$ be a CPTP map with Kraus rank $|\rE|\geq |\A|,|\B|$. Then, there exists a CP map $\widehat{\cN}:\cL(\A)\rightarrow\cL(\B)$ with Kraus rank at most $C\max(|\A|,|\B|)/\varepsilon^2$ (where $C>0$ is a universal constant) and such that
\[ \label{eq:approx'} \sup_{\varrho\in\cD(\A)} \left\| \widehat{\cN}(\varrho)-\cN(\varrho) \right\|_{\infty} \leq \varepsilon\sup_{\varrho\in\cD(\A)} \left\|\cN(\varrho)\right\|_{\infty}. \]
\end{theorem}

As a straightforward consequence of Theorem \ref{th:main'} we get the following: If $\cN:\cL(\A)\rightarrow\cL(\B)$ is such that $\|\cN(\varrho)\|_{\infty}\leq c/|\B|$ for all $\varrho\in\cD(\A)$, then $\|\widehat{\cN}(\varrho)-\cN(\varrho)\|_{\infty} \leq c\varepsilon/|\B|$ for all $\varrho\in\cD(\A)$, and hence $\|\widehat{\cN}(\varrho)-\cN(\varrho)\|_{1} \leq c\varepsilon$ for all $\varrho\in\cD(\A)$. This means that, if $\cN$ sends any input state to an output state that has small operator norm, then it can be approximated by a CP map with Kraus rank of order $\max(|\A|,|\B|)$ in $(1{\rightarrow}\infty)$-norm up to error of order $\varepsilon/|\B|$, and thus a fortiori in $(1{\rightarrow}1)$-norm up to error $\varepsilon$.

Before moving on to the full proof of Theorems \ref{th:main} and \ref{th:main'} let us briefly explain the main ideas in it. These two existence results of CPTP maps having some desired properties actually stem from proving that suitably constructed random ones have them with high probability. One thus has to show that for the random CPTP map $\widehat{\cN}$ the probability is high that, for every input state $\varrho$, $\widehat{\cN}(\varrho)$ is close to $\cN(\varrho)$. This is achieved in two steps: establishing first that this holds for a given input state $\varrho$ and second that it in fact holds for all of them simultaneously. The fact that the individual probability of deviating from average is small is a consequence of the concentration of measure phenomenon in high dimensions. Deriving from there that the global deviation probability is also small is done by discretizing the input set and using a union bound. This line of proof is extremely standard in asymptotic geometric analysis (this is for instance how Dvoretzky's theorem is obtained from Levy's lemma) or in large dimensional probability theory (this is for instance how the supremum of an empirical process is upper bounded through generic chaining). In our case though, the first step requires a careful analysis of the sub-exponential behavior of a certain random variable.

\section{Proofs of the main results}
\label{sec:proofs}

As a crucial step in establishing Theorems \ref{th:main} and \ref{th:main'}, we will need a large deviation inequality for sums of independent $\psi_1$ (aka sub-exponential) random variables. Recall that the $\psi_1$-norm of a random variable $X$ (which quantifies the exponential decay of the tail) may be defined via the growth of moments
\[ \|X\|_{\psi_1} = \sup_{p \in \N}\frac{\big(\E |X|^{p}\big)^{1/p}}{p} .\]
This definition is more practical than the standard definition through the Orlicz function $x \mapsto e^x-1$, and leads to an equivalent norm (see \cite{CGLP}, Corollary 1.1.6). The large deviation inequality for a sum of independent $\psi_1$ random variables is known as Bernstein's inequality and is quoted below.

\begin{theorem}[Bernstein's inequality, see e.g.~\cite{CGLP}, Theorem 1.2.5]
\label{th:Bernstein}
Let $X_1,\ldots,X_n$ be $n$ independent $\psi_1$ random variables.
Set $M=\max_{1\leq i\leq n}\|X_i\|_{\psi_1}$ and $\sigma^2=\sum_{1\leq i\leq n}\|X_i\|_{\psi_1}^2/n$. Then,
\[ \forall\ t>0,\ \P\left(\left|\frac{1}{n}\sum_{i=1}^n \left(X_i -\E X_i\right)\right| >t \right)\leq \exp\left(-c_0 n\min\left(\frac{t^2}{\sigma^2},\frac{t}{M}\right) \right),\]
where $c_0>0$ is a universal constant.
\end{theorem}

Our application of Bernstein's inequality to a suitably chosen sum of independent $\psi_1$ random variables will yield Proposition \ref{prop:fixed} below. Note that in the latter, as well as in several other places in the remainder of the paper, we shall use the following shorthand notation, whenever no confusion is at risk: given a unit vector $\phi$ in $\C^n$, we also denote by $\phi$ the corresponding pure state $\ket{\phi}\!\bra{\phi}$ on $\C^n$.

\begin{proposition} \label{prop:fixed}
Let $\cN:\cL(\A)\rightarrow\cL(\B)$ be a CPTP map with Kraus rank $|\rE|$, defined by
\begin{equation} \label{eq:N} \forall\ \varrho\in\cD(\A),\ \cN(\varrho)=\tr_{\rE}\left[V\varrho V^{\dagger}\right], \end{equation}
for some isometry $V:\A\hookrightarrow\B\otimes\rE$.

For any given unit vector $\varphi$ in $\rE$ define next the CP map $\cN_{\varphi}:\cL(\A)\rightarrow\cL(\B)$ by
\begin{equation} \label{eq:N_phi} \forall\ \varrho\in\cD(\A),\ \cN_{\varphi}(\varrho)= |\rE|\, \tr_{\rE}\left[\left(\Id\otimes\varphi\right) V\varrho V^{\dagger} \left(\Id\otimes\varphi\right) \right]. \end{equation}

Now, fix unit vectors $x$ in $\A$, $y$ in $\B$, and pick random unit vectors $\varphi_1,\ldots,\varphi_n$ in $\rE$, independently and uniformly. Then,
\[ \forall\ 0<\varepsilon<1,\ \P\left( \left|\frac{1}{n}\sum_{i=1}^n \bra{y}\cN_{\varphi_i}\left(x\right)\ket{y} - \bra{y}\cN\left(x\right)\ket{y} \right| > \varepsilon\bra{y}\cN\left(x\right)\ket{y} \right) \leq e^{-cn\varepsilon^2}, \]
where $c>0$ is a universal constant.
\end{proposition}

Note that, by construction, the CP map $\cN_{\varphi}:\cL(\A)\rightarrow\cL(\B)$ defined in Equation \eqref{eq:N_phi} has Kraus rank $1$. Indeed, it has $V_\varphi=\sqrt{|\rE|}\left(\Id\otimes\varphi\right) V:\A\hookrightarrow\B\otimes\mathrm{span}\{\varphi\}$ as Stinespring embedding, which has effective environment dimension $1$.

In order to derive this concentration result, we will need first of all an estimate on the $\psi_1$-norm of a certain random variable appearing in our construction. This is the content of Lemma \ref{lemma:psi_1} below.

\begin{lemma} \label{lemma:psi_1}
Fix $d,s\in\N$. Let $\sigma$ be a state on $\C^{d}\otimes\C^s$ and $y$ be a unit vector in $\C^{d}$. Next, for $\varphi$ a uniformly distributed unit vector in $\C^s$ define the random variable
\[ X_{\varphi}(\sigma,y) = \tr\left[ y\otimes\varphi\,\sigma\right]. \]
Then, $X_{\varphi}(\sigma,y)$ is a $\psi_1$ random variable with mean and $\psi_1$-norm satisfying
\begin{equation} \label{eq:mean-psi_1} \E X_{\varphi}(\sigma,y) = \frac{1}{s}\tr\left[ y\otimes\Id\,\sigma\right]\ \ \text{and}\ \ \|X_{\varphi}(\sigma,y)\|_{\psi_1} \leq \frac{1}{s}\tr\left[ y\otimes\Id\,\sigma\right]. \end{equation}
\end{lemma}

\begin{proof}
To begin with, recall that, for any $p\in\N$, we have, for $\varphi$ a uniformly distributed unit vector in $\C^s$,
\[ \E\varphi^{\otimes p} = \frac{1}{{s+p-1 \choose p}} P_{\Sym^{p}(\C^s)}, \]
where $P_{\Sym^{p}(\C^s)}$ denotes the orthogonal projector onto the completely symmetric subspace of $(\C^s)^{\otimes p}$. Indeed, $\E\varphi^{\otimes p}$ commutes with all $U^{\otimes p}$, for $U\in\mathcal{U}(\C^s)$, so by Schur's lemma it has to be proportional to $P_{\Sym^{p}(\C^s)}$, and the normalization is given by observing that $\tr(\E\varphi^{\otimes p})=1$, together with the fact that $\dim(\Sym^{p}(\C^s))={s+p-1 \choose p}$.

Now, setting $\sigma_y= \tr_{\C^{d}}\left[y\otimes\Id\, \sigma\right]$, positive sub-normalized operator on $\C^s$, we see that $X_{\varphi}(\sigma,y) = \tr\left[\varphi\,\sigma_y\right]$. Hence, we clearly have first of all the first statement in Equation \eqref{eq:mean-psi_1}, namely
\[ \E X_{\varphi}(\sigma,y) = \frac{1}{s} \tr \left[\Id\,\sigma_y\right] =\frac{1}{s} \tr\left[y\otimes\Id\,\sigma\right]. \]
What is more, for any $p\in\N$, $\left|X_{\varphi}(\sigma,y)\right|^p = \left(\tr \left[ \varphi\,\sigma_y\right]\right)^p = \tr \left[ \varphi^{\otimes p}\,\sigma_y^{\otimes p}\right]$. And therefore,
\[ \E \left|X_{\varphi}(\sigma,y)\right|^p = \frac{1}{{s+p-1 \choose p}} \tr\left[P_{\Sym^{p}(\C^s)}\sigma_y^{\otimes p}\right] \leq \frac{1}{{s+p-1 \choose p}} \tr\left[\sigma_y^{\otimes p}\right] \leq \left(\frac{p}{s} \tr\left[\sigma_y\right]\right)^p, \]
where the last inequality is simply by the rough bounds $p!\leq p^p$ and $(s+p-1)!/(s-1)!\geq s^p$.

So in the end, we get as wanted the second statement in Equation \eqref{eq:mean-psi_1}, namely
\[ \|X_{\varphi}(\sigma,y)\|_{\psi_1} =\sup_{p\in\N} \frac{\left(\E \left|X_{\varphi}(\sigma,y)\right|^p\right)^{1/p}}{p} \leq \frac{1}{s}\tr\left[y\otimes\Id\,\sigma\right]. \]
This concludes the proof of Lemma \ref{lemma:psi_1}.
\end{proof}

\begin{proof}[Proof of Proposition \ref{prop:fixed}] Note first of all that we can obviously re-write
\[ \bra{y}\cN(x)\ket{y} = \tr\left[y\otimes\Id\, VxV^{\dagger}\right]\ \text{and}\ \forall\ \varphi\in S_{\rE},\ \bra{y}\cN_{\varphi}(x)\ket{y} = |\rE|\,\tr\left[y\otimes\varphi\, VxV^{\dagger}\right]. \]
Next, for each $1\leq i\leq n$, define the random variable $Y_i=\bra{y}\cN_{\varphi_i}(x)\ket{y}$. By Lemma \ref{lemma:psi_1}, combined with the observation just made above, we know that these are independent $\psi_1$ random variables with mean $\bra{y}\cN(x)\ket{y}$ and $\psi_1$-norm upper bounded by $\bra{y}\cN(x)\ket{y}$. So by Bernstein's inequality, recalled as Theorem \ref{th:Bernstein}, we get that
\[ \forall\ t>0,\ \P\left( \left|\frac{1}{n}\sum_{i=1}^n Y_i - \bra{y}\cN(x)\ket{y} \right| > t \right) \leq \exp\left(-c_0n\min\left(\frac{t^2}{\bra{y}\cN(x)\ket{y}^2},\frac{t}{\bra{y}\cN(x)\ket{y}}\right) \right), \]
where $c_0>0$ is a universal constant. And hence,
\[ \forall\ 0<\varepsilon<1,\ \P\left( \left|\frac{1}{n}\sum_{i=1}^n Y_i - \bra{y}\cN(x)\ket{y} \right| > \varepsilon \bra{y}\cN(x)\ket{y} \right) \leq e^{-c_0n\varepsilon^2}, \]
which is precisely the result announced in Proposition \ref{prop:fixed}.
\end{proof}

Having at hand the ``fixed $x,y$'' concentration inequality of Proposition \ref{prop:fixed}, we can now get its ``for all $x,y$'' counterparts by a standard net-argument. This appears as the following Propositions \ref{prop:forall} and \ref{prop:forall'}.

\begin{proposition} \label{prop:forall}
	Let $\cN:\cL(\A)\rightarrow\cL(\B)$ be a CPTP map, as characterized by Equation \eqref{eq:N}, and for each unit vector $\varphi$ in $\rE$ define the CP map $\cN_{\varphi}:\cL(\A)\rightarrow\cL(\B)$ as in Equation \eqref{eq:N_phi}. Next, for $\varphi_1,\ldots,\varphi_n$ independent uniformly distributed unit vectors in $\rE$, set $\cN_{\varphi^{(n)}}=\big(\sum_{i=1}^n\cN_{\varphi_i}\big)/n$. Then, for any $0<\varepsilon<1$,
	\begin{align*} 
	& \P\left( \forall\ x\in S_{\A},y\in S_{\B},\ \left|\bra{y}\cN_{\varphi^{(n)}}(x)-\cN(x)\ket{y} \right| \leq \varepsilon\bra{y}\cN(x)\ket{y} + \frac{\varepsilon}{|\B|} \right) \\
	& \qquad \geq 1 - \left(\frac{24|\rE||\B|}{\varepsilon}\right)^{2(|\A|+|\B|)} e^{-cn\varepsilon^2}, \end{align*}
	where $c>0$ is a universal constant.
\end{proposition}

Note that, by construction, the random CP map $\cN_{\varphi^{(n)}}:\cL(\A)\rightarrow\cL(\B)$ introduced above has Kraus rank at most $n$. Indeed, it is a convex combination of $n$ Kraus rank $1$ random CP maps $\cN_{\varphi_i}:\cL(\A)\rightarrow\cL(\B)$, $1\leq i\leq n$.

\begin{proof}
	Fix $0<\alpha,\beta<1$ and consider $\mathcal{A}_{\alpha},\mathcal{B}_{\beta}$ minimal $\alpha,\beta$-nets within the unit spheres of $\A,\B$, so that by a standard volumetric argument $\left|\mathcal{A}_{\alpha}\right|\leq (3/\alpha)^{2|\A|},\left|\mathcal{B}_{\beta}\right|\leq (3/\beta)^{2|\B|}$ (see e.g.~\cite{Pisier}, Chapter 4). Then, by Proposition \ref{prop:fixed} and the union bound, we get that, for any $\varepsilon>0$,
	\begin{align} \label{eq:M_delta} 
	& \P\left( \forall\ x\in\mathcal{A}_{\alpha},y\in\mathcal{B}_{\beta},\ \left| \bra{y} \cN_{\varphi^{(n)}}(x)-\cN(x) \ket{y} \right| \leq \varepsilon\bra{y}\cN\left(x\right)\ket{y} \right) \nonumber \\
	&\qquad \geq 1 - \left(\frac{3}{\alpha}\right)^{2|\A|}\left(\frac{3}{\beta}\right)^{2|\B|}e^{-cn\varepsilon^2}. 
	\end{align}
	Now, fix $\varepsilon>0$ and suppose that $\mathcal{E}:\cL(\A)\rightarrow\cL(\B)$ is a Hermiticity-preserving map which is such that
	\begin{equation} \label{eq:true-net} \forall\ x\in\mathcal{A}_{\alpha},\ \forall\ y\in\mathcal{B}_{\beta},\ \left|\bra{y}\mathcal{E}(x)\ket{y}\right| \leq \varepsilon\bra{y}\cN\left(x\right)\ket{y}. \end{equation}
	Assume that $\mathcal{E}$ additionally satisfies the boundedness property
	\begin{equation} \label{eq:boundedness} \forall\ x\in S_{\A},\ \forall\ y\in S_{\B},\ \left|\bra{y}\mathcal{E}(x)\ket{y}\right| \leq |\rE|.\end{equation}
	Note that if $\mathcal{E}$ is Hermicity-preserving, then 
	\begin{align*}
	    & \forall\ y\in S_{\B},\ \sup_{v,v'\in S_{\A}}|\bra{y}\mathcal{E}(\ketbra{v}{v'})\ket{y}|= \sup_{v\in S_{\A}}|\bra{y}\mathcal{E}(\ketbra{v}{v})\ket{y}| ,\\
	    & \forall\ x\in S_{\A},\ \sup_{w,w'\in S_{\B}}|\bra{w}\mathcal{E}(\ketbra{x}{x})\ket{w'}|= \sup_{w\in S_{\B}}|\bra{w}\mathcal{E}(\ketbra{x}{x})\ket{w}| .
	\end{align*}
	Indeed, this is because for any $X$, $\mathcal{E}(X^{\dagger})=\mathcal{E}(X)^{\dagger}$. Hence, it will be useful to us later on to keep in mind that assumption \eqref{eq:boundedness} is actually equivalent to
	\[ \forall\ x,x'\in S_{\A},\ \forall\ y,y'\in S_{\B},\ \begin{cases} \left|\bra{y}\mathcal{E}(\ketbra{x}{x'})\ket{y}\right| \leq |\rE| \\ \left|\bra{y}\mathcal{E}(\ketbra{x}{x})\ket{y'}\right| \leq |\rE| \end{cases}. \]
	Then, for any unit vectors $x\in S_{\A}$, $y\in S_{\B}$, we know by definition that there exist $\tilde{x}\in\mathcal{A}_{\alpha}$, $\tilde{y}\in\mathcal{B}_{\beta}$ such that $\|x-\tilde{x}\|\leq\alpha$, $\|y-\tilde{y}\|\leq\beta$. Hence, first of all
	\begin{align*} \left|\bra{y}\mathcal{E}(\proj{x})\ket{y}\right| & \leq \left|\bra{\tilde{y}}\mathcal{E}(\proj{x})\ket{\tilde{y}}\right| + \left|\bra{y-\tilde{y}}\mathcal{E}(\proj{x})\ket{\tilde{y}}\right| + \left|\bra{y}\mathcal{E}(\proj{x})\ket{y-\tilde{y}}\right|\\
	& \leq \left|\bra{\tilde{y}}\mathcal{E}(\proj{x})\ket{\tilde{y}}\right| + 2\beta|\rE|, \end{align*}
	where the second inequality follows from the boundedness property \eqref{eq:boundedness} of $\mathcal{E}$, combined with the fact that $\|y-\tilde{y}\|\leq\beta$. Then similarly, because $\|x-\tilde{x}\|\leq\alpha$,
	\begin{align*} \left|\bra{\tilde{y}}\mathcal{E}(\proj{x})\ket{\tilde{y}}\right| & \leq \left|\bra{\tilde{y}}\mathcal{E}(\proj{\tilde{x}})\ket{\tilde{y}}\right| + \left|\bra{\tilde{y}}\mathcal{E}(\ketbra{x-\tilde{x}}{\tilde{x}})\ket{\tilde{y}}\right| + \left|\bra{\tilde{y}}\mathcal{E}(\ketbra{x}{x-\tilde{x}})\ket{\tilde{y}}\right| \\
	& \leq \left|\bra{\tilde{y}}\mathcal{E}(\proj{\tilde{x}})\ket{\tilde{y}}\right| + 2\alpha|\rE|. \end{align*}
	Putting together the two previous upper bounds, we see that we actually have
	\[ \left|\bra{y}\mathcal{E}(\proj{x})\ket{y}\right| \leq \left|\bra{\tilde{y}}\mathcal{E}(\proj{\tilde{x}})\ket{\tilde{y}}\right| + 2|\rE|(\alpha+\beta) \leq \varepsilon \bra{\tilde{y}}\mathcal{N}(\proj{\tilde{x}})\ket{\tilde{y}} + 2|\rE|(\alpha+\beta), \]
	where the second inequality is by assumption \eqref{eq:true-net} on $\mathcal{E}$. Now, arguing just as before (using this time that $\mathcal{N}$ satisfies the boundedness property $\left|\bra{y}\mathcal{N}(\ketbra{x}{x'})\ket{y'}\right|\leq 1$ for any $x,x'\in S_{\A}$ and $y,y'\in S_{\B}$), we get
	\[ \bra{\tilde{y}}\mathcal{N}(\proj{\tilde{x}})\ket{\tilde{y}} \leq \bra{y}\mathcal{N}(\proj{x})\ket{y}+2(\alpha+\beta) . \]
	So eventually, what we obtain is
	\[ \left|\bra{y}\mathcal{E}(x)\ket{y}\right| \leq \varepsilon\big(\bra{y}\cN(x)\ket{y} + 2(\alpha+\beta)\big) + 2|\rE|(\alpha+\beta) \leq \varepsilon\bra{y}\cN(x)\ket{y} + 4|\rE|(\alpha+\beta). \]
	Therefore, choosing $\alpha=\beta=\varepsilon/(8|\rE||\B|)$ (and observing that, by the way $\cN_{\varphi^{(n)}}$ is constructed, $\cN_{\varphi^{(n)}}-\cN$ fulfills condition \eqref{eq:boundedness}), it follows from Equation \eqref{eq:M_delta} that, for any $\varepsilon>0$,
	\begin{align*} 
	& \P\left( \forall\ x\in S_{\A},y\in S_{\B},\ \left| \bra{y}\cN_{\varphi^{(n)}}(x)-\cN(x)\ket{y} \right| \leq \varepsilon\bra{y}\cN\left(x\right)\ket{y} +\frac{\varepsilon}{|\B|} \right) \\
	& \qquad \geq 1 - \left(\frac{24|\rE||\B|}{\varepsilon}\right)^{2(|\A|+|\B|)}e^{-cn\varepsilon^2}, 
	\end{align*}
	which is exactly what we wanted to show.
\end{proof}

\begin{proposition} \label{prop:forall'}
	Let $\cN:\cL(\A)\rightarrow\cL(\B)$ be a CPTP map, as characterized by Equation \eqref{eq:N}, and for each unit vector $\varphi$ in $\rE$ define the CP map $\cN_{\varphi}:\cL(\A)\rightarrow\cL(\B)$ as in Equation \eqref{eq:N_phi}. Next, for $\varphi_1,\ldots,\varphi_n$ independent uniformly distributed unit vectors in $\rE$, set $\cN_{\varphi^{(n)}}=\big(\sum_{i=1}^n\cN_{\varphi_i}\big)/n$. Then, for any $0<\varepsilon<1$,
	\[ \P\left( \sup_{x\in S_{\A},y\in S_{\B}} \left| \bra{y}\cN_{\varphi^{(n)}}(x)-\cN(x)\ket{y} \right| \leq \varepsilon \sup_{x\in S_{\A},y\in S_{\B}} \bra{y}\cN(x)\ket{y} \right) \geq 1 - 225^{|\A|+|\B|} e^{-cn\varepsilon^2}, \]
	where $c>0$ is a universal constant.
\end{proposition}

\begin{proof}
	We will argue in a way very similar to what was done in the proof of Proposition \ref{prop:forall}, and hence skip some of the details here. Again, fix $0<\alpha,\beta< 1/4$ and consider $\mathcal{A}_{\alpha},\mathcal{B}_{\beta}$ minimal $\alpha,\beta$-nets within the unit spheres of $\A,\B$.
	Now, fix $\varepsilon>0$ and suppose that $\mathcal{E}:\cL(\A)\rightarrow\cL(\B)$ is a Hermiticity-preserving map which is such that,
	\[ \forall\ x\in\mathcal{A}_{\alpha},\ \forall\ y\in\mathcal{B}_{\beta},\ \left|\bra{y}\mathcal{E}(x)\ket{y}\right| \leq \varepsilon\bra{y}\cN\left(x\right)\ket{y}. \]
	Then, for any unit vectors $x\in S_{\A}$, $y\in S_{\B}$,
	\begin{align*}
	& \left|\bra{y}\mathcal{E}(\proj{x})\ket{y}\right| \\
	& \qquad \leq \varepsilon \big( \bra{y}\cN(\proj{x})\ket{y} + 2\alpha\,{\sup}_{v,v'}\bra{y}\cN(\ketbra{v}{v'})\ket{y} + 2\beta\,{\sup}_{w,w'} \bra{w}\cN(\proj{\tilde{x}})\ket{w'} \big) \\
	& \qquad \quad + 2\alpha\,{\sup}_{v,v'}\left|\bra{\tilde{y}}\mathcal{E}(\ketbra{v}{v'})\ket{\tilde{y}}\right| + 2\beta\,{\sup}_{w,w'}\left|\bra{w}\mathcal{E}(\proj{x})\ket{w'}\right|,
	\end{align*}
	where $\tilde{x}\in\mathcal{A}_{\alpha}$, $\tilde{y}\in\mathcal{B}_{\beta}$ are such that $\|x-\tilde{x}\|\leq\alpha$, $\|y-\tilde{y}\|\leq\beta$. And consequently, taking supremum over unit vectors $x\in S_{\A}$, $y\in S_{\B}$, we get
	\[ {\sup}_{x,y}\left|\bra{y}\mathcal{E}(x)\ket{y}\right| \leq\,  \varepsilon(1+2(\alpha+\beta))\,{\sup}_{x,y} \bra{y}\cN(x)\ket{y} + 2(\alpha+\beta) \,{\sup}_{x,y}\left|\bra{y}\mathcal{E}(x)\ket{y}\right|, \]
	that is equivalently,
	\[ {\sup}_{x,y}\left|\bra{y}\mathcal{E}(x)\ket{y}\right| \leq \varepsilon\,\frac{1+2(\alpha+\beta)}{1-2(\alpha+\beta)}\,{\sup}_{x,y} \bra{y}\cN(x)\ket{y}. \]
	Therefore, choosing $\alpha=\beta=1/5$, so that $(1+2(\alpha+\beta))/(1-2(\alpha+\beta))=9$ and $3/\alpha=3/\beta=15$, we eventually obtain that, for any $0<\varepsilon<1$,
	\[ \P\left( \sup_{x\in S_{\A},y\in S_{\B}} \left| \bra{y}\cN_{\varphi^{(n)}}(x)-\cN(x)\ket{y} \right| \leq 9\varepsilon \sup_{x\in S_{\A},y\in S_{\B}} \bra{y}\cN(x)\ket{y} \right) \geq 1 - 15^{2(|\A|+|\B|)}e^{-cn\varepsilon^2}, \]
	which, after relabelling $9\varepsilon$ in $\varepsilon$, implies precisely the result announced in Proposition \ref{prop:forall'}.
\end{proof}

\begin{proof}[Proof of Theorem \ref{th:main}]
	Because operator-ordering is preserved by convex combinations, it follows from Proposition \ref{prop:forall} that there exists a universal constant $c>0$ such that, for any $\varepsilon>0$,
	\begin{align*} 
	& \P\left( \forall\ \rho\in\cD(\A),\ -\varepsilon\left(\cN(\rho) +\frac{\Id}{|\B|}\right) \leq \cN_{\varphi^{(n)}}(\rho) -\cN(\rho) \leq \varepsilon\left(\cN(\rho) +\frac{\Id}{|\B|}\right) \right) \\
	& \qquad \geq 1 - \left(\frac{24|\rE||\B|}{\varepsilon}\right)^{2(|\A|+|\B|)}e^{-cn\varepsilon^2}. 
	\end{align*}
	The r.h.s.~of the latter inequality becomes larger than, say, $1/2$ as soon as $n$ is larger than $C\max(|\A|,|\B|)\log (|\rE|/\varepsilon)/\varepsilon^2$, for $C>0$ some universal constant.
	
	Recapitulating, what we have shown so far is that there exists a completely positive map $\cN^{(n)}$ with Kraus rank $n\leq C\max(|\A|,|\B|)\log (|\rE|/\varepsilon)/\varepsilon^2$, for $C>0$ some universal constant, such that,
	\begin{equation} \label{eq:CP} \forall\ \rho\in\cD(\A),\ -\varepsilon\left(\cN(\rho) +\frac{\Id}{|\B|}\right) \leq \cN^{(n)}(\rho)-\cN(\rho)\leq \varepsilon\left(\cN(\rho) +\frac{\Id}{|\B|}\right). \end{equation}
	In particular, Equation \eqref{eq:CP} implies that, for any $\rho\in\cD(\A)$, $\left|\tr\left(\cN^{(n)}(\rho)\right) -1\right|\leq 2\varepsilon$, so that $\cN^{(n)}$ is almost trace preserving, up to an error $2\varepsilon$. As a consequence of Equation \eqref{eq:CP}, we also have
	\begin{equation} \label{eq:1-norm} \forall\ \rho\in\cD(\A),\ \left\|\cN^{(n)}(\rho)-\cN(\rho)\right\|_1 \leq 2\varepsilon, \end{equation}
	and to get only such trace-norm approximation, it is actually possible to impose that $\cN^{(n)}$ is strictly trace preserving. Indeed, denote by $\{K_1,\ldots,K_n\}$ a set of Kraus operators for $\cN^{(n)}$, and set $S=\sum_{i=1}^nK_i^{\dagger}K_i$. Equation \eqref{eq:1-norm} guarantees that $\|S-\Id\|_{\infty}\leq 2\varepsilon$, so that $S$ is in particular invertible, as soon as $\varepsilon<1/2$. 
	Hence, consider the completely positive map $\widehat{\cN}^{(n)}$ having $\{K_1S^{-1/2},\ldots,K_nS^{-1/2}\}$ as a set of Kraus operators, which means that $\widehat{\cN}^{(n)}(\cdot)=\cN^{(n)}(S^{-1/2}\,\cdot\, S^{-1/2})$. The latter is trace preserving by construction, and such that
	\begin{equation} \label{eq:TP-approx} \forall\ \rho\in\cD(\A),\ \left\|\widehat{\cN}^{(n)}(\rho)-\cN(\rho)\right\|_1\leq \left\|\widehat{\cN}^{(n)}(\rho)-\cN^{(n)}(\rho)\right\|_1 + \left\|\cN^{(n)}(\rho)-\cN(\rho)\right\|_1 \leq 8\varepsilon. \end{equation}
	Indeed, for any $\rho\in\cD(\A)$, we have the chain of inequalities
	\[ \|\rho-S^{1/2}\rho S^{1/2}\|_1\leq \left(\|S^{1/2}\|_{\infty}+\|\Id\|_{\infty}\right) \|\rho\|_1 \|S^{1/2}-\Id\|_{\infty} \leq (1+\varepsilon+1)\,2\varepsilon\leq 6\varepsilon, \]
	where the first inequality follows from the triangle and H\"{o}lder inequalities (after simply noticing that, setting $\Delta=\Id-S^{1/2}$, we can rewrite $\rho-S^{1/2}\rho S^{1/2}$ as $\Delta\rho\Id + S^{1/2}\rho\Delta$), while the second inequality is because, for any $0<x<1/2$, $(1+2x)^{1/2}\leq 1+x$ and $(1-2x)^{1/2}\geq 1-2x$, so that $\|S^{1/2}\|_{\infty}\leq 1+\varepsilon$ and $\|\Id-S^{1/2}\|_{\infty}\leq 2\varepsilon$. This implies that, for any $\rho\in\cD(\A)$,
	\begin{equation} \label{eq:contraction} \left\|\widehat{\cN}^{(n)}(\rho)-\cN^{(n)}(\rho)\right\|_1 = \left\|\widehat{\cN}^{(n)}\left(\rho-S^{1/2}\rho S^{1/2}\right)\right\|_1 \leq \left\|\rho-S^{1/2}\rho S^{1/2}\right\|_1 \leq 6\varepsilon, 
	\end{equation}
	where the first inequality is because $\widehat{\cN}^{(n)}$ is a CPTP map, and hence has $(1\to 1)$-norm equal to $1$.
	Combining \eqref{eq:contraction} and \eqref{eq:1-norm}, we get the last inequality in \eqref{eq:TP-approx}.
    
	This concludes the proof of Theorem \ref{th:main} and of Remark \ref{remark:schatten} following it.
\end{proof}

\begin{proof}[Proof of Theorem \ref{th:main'}]
	By extremality of pure states amongst all states, it follows from Proposition \ref{prop:forall'} that there exists a universal constant $c>0$ such that, for any $\varepsilon>0$,
	\[ \P\left( \sup_{\rho\in\cD(\A)} \left\| \cN_{\varphi^{(n)}}(\rho) -\cN(\rho) \right\|_{\infty} \leq \varepsilon \sup_{\rho\in\cD(\A)} \left\|\cN(\rho)\right\|_{\infty} \right) \geq 1 - 225^{|\A|+|\B|}e^{-cn\varepsilon^2}. \]
	The r.h.s.~of the latter inequality becomes larger than, say, $1/2$ as soon as $n$ is larger than $C\max(|\A|,|\B|)/\varepsilon^2$, for $C>0$ some universal constant. And the proof of Theorem \ref{th:main'} is thus complete.
\end{proof}

\begin{remark} \label{rem:Kraus}
The way we construct a random quantum channel approximating with high probability a quantum channel of interest is by starting from its Stinespring representation. Indeed, as briefly explained in the introduction, starting from one of its Kraus representations is a priori not as convenient, because of their non uniqueness. However, our random construction can be rephrased in terms of Kraus operators, in a way that does not depend on the chosen representation. The correspondence is as follows: Assume that our CPTP map $\mathcal{N}:\cL(\A)\rightarrow\cL(\B)$ can be written in the Stinespring and Kraus pictures, respectively, as
\[ \mathcal{N}(X) = \tr_{\rE}\left(VXV^{\dagger}\right) = \sum_{i=1}^{|\rE|} K_iXK_i^{\dagger}. \]
Then, given a unit vector $\varphi\in\rE$, we define in Equation \eqref{eq:N_phi} the CP map $\cN_{\varphi}:\cL(\A)\rightarrow\cL(\B)$ by
\[ \cN_{\varphi}(X)= \tr_{\rE}\left(V_{\varphi} X V_{\varphi}^{\dagger} \right),\ \text{where}\ V_{\varphi}=\sqrt{|\rE|}\left(\Id\otimes\varphi\right) V. \] 
Now, the latter can be equivalently defined by
\[ \cN_{\varphi}(X)= K_{\varphi} X K_{\varphi}^{\dagger},\ \text{where}\ K_{\varphi}=\sqrt{|\rE|}\sum_{i=1}^{|\rE|} \varphi_iK_i. \]
Hence, sampling Kraus operators at random from the set $\{K_1,\ldots,K_{|\rE|}\}$ does not work in general (and it is not even clear how to chose a representation for which it would), but sampling them as random weighted sums from the set $\{K_1,\ldots,K_{|\rE|}\}$ does work (whatever the chosen representation).

This is in contrast with the case of the fully randomizing channel, studied in \cite{HLSW} and \cite{Aubrun}. For the latter there is a well-identified distribution to sample Kraus operators from (namely Haar-distributed unitaries), and even a distinguished Kraus decomposition to directly sub-sample Kraus operators from (namely that built from generalized Pauli shift and phase operators). 
\end{remark}

\section{Consequences and applications}
\label{sec:applications}

\subsection{Approximation in terms of output entropies or fidelities} \hfill\smallskip

This section gathers some (more or less straightforward) corollaries of Theorem \ref{th:main} concerning approximation of quantum channels in other distance measures than the $(1{\rightarrow}1)$-norm distance mostly studied up to now.

Given a state $\varrho$ on some Hilbert space $\H$, we define, for any $p\in]1,\infty[$, its R\'{e}nyi entropy of order $p$ as
\[ S_p(\varrho)=-\frac{p}{p-1}\log\|\varrho\|_p, \]
and the latter definition is extended by continuity to $p\in\{1,\infty\}$ as
\[ S_1(\varrho)=S(\varrho)=-\tr(\varrho\log\varrho)\ \text{and}\ S_{\infty}(\varrho)=-\log\lambda_{\max}(\varrho). \]
R\'{e}nyi $p$-entropies thus measure the amount of information present in a quantum state, generalizing the case $p=1$ of the von Neumann entropy. Besides, given states $\rho,\sigma$ on some Hilbert space $\H$, their fidelity is defined as $F(\rho,\sigma) = \|\sqrt{\rho}\sqrt{\sigma}\|_1$.

Now, given a channel $\cN$, from some input Hilbert space $\A$ to some output Hilbert space $\B$, it is important to understand quantities such as its minimum output R\'{e}nyi $p$-entropy, i.e.
\[ S_p^{\min}(\cN)=\min_{\rho\in\cD(\A)} S_p\big(\cN(\rho)\big), \] 
or its maximum output fidelity with a fixed state $\sigma$ on $\B$, i.e.
\[ F^{\max}(\cN,\sigma)=\max_{\rho\in\cD(\A)} F\big(\cN(\rho),\sigma\big). \] 
Indeed, in quantum Shannon theory, these are relevant for their own sake in the asymptotic memoryless setting, while smoothed versions of them show up in the one-shot setting.
It is thus of interest to have a channel $\widehat{\cN}$ which is less complex than $\cN$ but nevertheless shares approximately the same $S_p^{\min}$ and $F^{\max}(\cdot,\sigma)$.
What is more, as we will see shortly, channel approximation results in terms of output von Neumann entropy (and perhaps also in terms of other entropies or fidelities) could be the key to establishing channel-dependent lower bounds on the achievable Kraus rank reduction.

\begin{proposition} \label{prop:S_p}
Let $\cN:\cL(\A)\rightarrow\cL(\B)$ be a CPTP map, and assume that the CP map $\widehat{\cN}:\cL(\A)\rightarrow\cL(\B)$ satisfies
\begin{equation} \label{eq:order} \forall\ \varrho\in\cD(\A),\ (1-\varepsilon)\cN(\varrho) -\varepsilon\frac{\Id}{|\B|} \leq \widehat{\cN}(\varrho)\leq (1+\varepsilon)\cN(\varrho) +\varepsilon\frac{\Id}{|\B|}, \end{equation}
for some $0<\varepsilon<1/2$. Then, for any $p\in]1,\infty]$, $\widehat{\cN}$ is close to $\cN$ in terms of output $p$-entropies, in the sense that
\[ \forall\ \varrho\in\cD(\A),\ \left| S_p\big(\widehat{\cN}(\varrho)\big)-S_p\big(\cN(\varrho)\big) \right| \leq \frac{p}{p-1}4\varepsilon. \]
\end{proposition}

\begin{proof}
Setting $\sigma=\cN(\varrho)$, $\widehat{\sigma}=\widehat{\cN}(\varrho)$ and $\tau=\Id/|\B|$, we can re-write Equation \eqref{eq:order} as the two inequalities
\[ \widehat{\sigma}\leq (1+\varepsilon)\sigma +\varepsilon\tau\ \text{and}\ \sigma\leq \frac{1}{1-\varepsilon}\widehat{\sigma} +\frac{\varepsilon}{1-\varepsilon}\tau \leq (1+2\varepsilon)\widehat{\sigma} +2\varepsilon\tau. \]
By operator monotonicity and the triangle inequality for $\|\cdot\|_p$, these imply the two estimates
\begin{equation} \label{eq:p} \|\widehat{\sigma}\|_p \leq (1+\varepsilon)\|\sigma\|_p +\varepsilon\|\tau\|_p\ \text{and}\ \|\sigma\|_p \leq (1+2\varepsilon)\|\widehat{\sigma}\|_p +2\varepsilon\|\tau\|_p. \end{equation}
Now, from the first inequality in Equation \eqref{eq:p}, we get
\[ \log\|\widehat{\sigma}\|_p \leq \log\|\sigma\|_p + \log(1+\varepsilon) + \frac{\varepsilon}{1+\varepsilon}\frac{\|\tau\|_p}{\|\sigma\|_p} \leq \log\|\sigma\|_p + 2\varepsilon, \]
where we used first that $\log$ is non-decreasing, then twice that $\log(1+x)\leq x$, and finally that $\|\sigma\|_p\leq \|\tau\|_p$. Similarly, we derive from the second inequality in Equation \eqref{eq:p} that
\[ \log\|\sigma\|_p  \leq \log\|\widehat{\sigma}\|_p + 4\varepsilon. \]
Multiplying the two previous inequalities by $-p/(p-1)<0$, we eventually obtain
\[ S_p(\widehat{\sigma}) \geq S_p(\sigma) -\frac{p}{p-1}2\varepsilon\ \text{and}\ S_p(\sigma) \geq S_p(\widehat{\sigma}) -\frac{p}{p-1}4\varepsilon. \]
The conclusion of Proposition \ref{prop:S_p} then follows.
\end{proof}

\begin{proposition} \label{prop:S}
Let $\cN:\cL(\A)\rightarrow\cL(\B)$ be a CPTP map, and assume that the CP map $\widehat{\cN}:\cL(\A)\rightarrow\cL(\B)$ satisfies
\begin{equation} \label{eq:approx} \forall\ \varrho\in\cD(\A),\ \left\|\widehat{\cN}(\varrho)-\cN(\varrho)\right\|_1 \leq \frac{2\varepsilon}{\log|\B|}, \end{equation}
for some $0<\varepsilon<1/2$. Then, $\widehat{\cN}$ is close to $\cN$ in terms of output entropies, in the sense that
\[ \forall\ \varrho\in\cD(\A),\ \left| S\big(\widehat{\cN}(\varrho)\big)-S\big(\cN(\varrho)\big) \right| \leq 4\sqrt{\varepsilon}. \]
\end{proposition}

\begin{proof}
By Fannes-Audenaert inequality \cite{Audenaert} (see also \cite{AW:entropy} for 
a streamlined proof), Equation \eqref{eq:approx} implies that
\[ \left| S\big(\widehat{\cN}(\varrho)\big) - S\big(\cN(\varrho)\big) \right| \leq \varepsilon - \frac{\varepsilon}{\log|\B|}\log \left(\frac{\varepsilon}{\log|\B|}\right) - \left(1-\frac{\varepsilon}{\log|\B|}\right)\log\left(1-\frac{\varepsilon}{\log|\B|}\right). \]
Now, for any $0<x<1/2$, on the one hand $x\log (1/x)\leq \sqrt{x}$, and on the other hand $\log(1/(1-x))\leq\log(1+2x)\leq 2x$ so that $(1-x)\log(1/(1-x))\leq 2x$. Hence,
\[ \left| S\big(\widehat{\cN}(\varrho)\big) - S\big(\cN(\varrho)\big) \right| \leq \varepsilon + \frac{2\varepsilon}{\log|\B|} + \sqrt{\frac{\varepsilon}{\log|\B|}}, \]
and the conclusion of Proposition \ref{prop:S} follows.
\end{proof}

\begin{theorem} \label{th:entropies}
Fix $0<\varepsilon<1$ and let $\cN:\cL(\A)\rightarrow\cL(\B)$ be a CPTP map with Kraus rank $|\rE|\geq |\A|,|\B|$. Then, there exists a CP map $\widehat{\cN}:\cL(\A)\rightarrow\cL(\B)$ with Kraus rank at most $C\max(|\A|,|\B|)\log(|\rE|/\varepsilon)/\varepsilon^2$ (where $C>0$ is a universal constant) and such that
\[ \forall\ p\in]1,\infty],\ \forall\ \varrho\in\cD(\A),\  \left| S_p\big(\widehat{\cN}(\varrho)\big)-S_p\big(\cN(\varrho)\big) \right| \leq \frac{p}{p-1}\varepsilon. \]
Besides, there also exists a CPTP map $\widehat{\cN}':\cL(\A)\rightarrow\cL(\B)$ with Kraus rank at most $C\max(|\A|,|\B|)\log^5(|\rE|/\varepsilon^2)/\varepsilon^4$ (where $C>0$ is a universal constant) and such that
\begin{equation} \label{eq:approx-S} \forall\ \varrho\in\cD(\A),\ \left| S\big(\widehat{\cN}'(\varrho)\big)-S\big(\cN(\varrho)\big) \right| \leq \varepsilon. \end{equation}
\end{theorem}

\begin{proof}
This is a direct consequence of Theorem \ref{th:main}, combined with Propositions \ref{prop:S_p} and \ref{prop:S}.
\end{proof}

We already argued about optimality in Theorem \ref{th:main}, showing that there exist CPTP maps $\cN:\cL(\A)\rightarrow\cL(\B)$ for which at least $|\A|$ or $|\B|$ Kraus operators are needed to approximate them in the sense of Equation \eqref{eq:approxmain}. We will now establish that, even to get the weaker notion of approximation of Equation \eqref{eq:approx-S}, a Kraus rank of at least $|\A|$ or $|\B|$ might, in some cases, still be necessary.

Let $\cN:X\in\cL(\A)\mapsto\tr_{\mathrm{E}}(VX V^{\dagger})\in\cL(\B)$ be a CPTP map with isometry $V:\A\hookrightarrow\B\otimes\rE$. Given $\varrho\in\cD(\A)$, we consider its input entropy $S(\varrho)$, its output entropy $S\left(\cN(\varrho)\right)$, and its entropy exchange $S\left(\varrho,\cN\right)$. The latter quantity is defined as follows: let $\varphi_{A'A}$ be an extension of $\varrho_A$, $\widetilde{\varphi}_{A'BE}=\left(\Id_{A'}\otimes V_{A\rightarrow BE}\right)\varphi_{A'A}$, and set
\[ S\left(\varrho_A,\cN_{A\rightarrow B}\right)= S\left(\tr_{\rE}\widetilde{\varphi}_{A'BE}\right)= S\left(\tr_{\A'\B}\widetilde{\varphi}_{A'BE}\right). \]
By non-negativity of the loss and the noise of a quantum channel, we then have (see \cite{GH}, Section 4.5)
\[ \forall\ \varrho\in\cD(\A),\ \left|S(\varrho)-S(\cN(\varrho))\right| \leq S(\varrho,\cN). \]
Yet, for any $\varrho\in\cD(\A)$, obviously $S(\varrho,\cN)\leq\log|\rE|$. And hence as a consequence,
\[ \log|\rE| \geq \max\left\{ \left|S(\varrho)-S(\cN(\varrho))\right| \st \varrho\in\cD(\A) \right\}. \]
In particular, we may derive the two following lower bounds on $|\rE|$, for certain CPTP maps $\cN$,
\begin{equation} \label{eq:S-exchange} \exists\ \psi_{\A}:\ \cN\left(\psi_A\right)=\frac{\Id_B}{|\B|}\ \Longrightarrow\ |\rE|\geq|\B|\ \text{and}\ \exists\ \psi_B:\ \cN\left(\frac{\Id_A}{|\A|}\right)=\psi_B\ \Longrightarrow\ |\rE|\geq|\A|\ . \end{equation}
And this remains approximately true for an approximation of $\cN$. Concretely, let $\widehat{\cN}:\cL(\A)\rightarrow\cL(\B)$ be a CPTP map such that
\[ \forall\ \varrho\in\cD(\A),\ \left| S\big(\widehat{\cN}(\varrho)\big)-S\big(\cN(\varrho)\big) \right| \leq \varepsilon. \]
If $\cN$ satisfies the first condition in Equation \eqref{eq:S-exchange}, then
\[ S\big(\widehat{\cN}\left(\psi_A\right)\big) \geq S\big(\cN\left(\psi_A\right)\big) -\varepsilon = \log|\B| - \varepsilon,\]
which implies that 
\[ \log r_K\big(\widehat{\cN}\big) \geq \log|\B|-\varepsilon,\ \text{i.e.}\ r_K\big(\widehat{\cN}\big) \geq e^{-\varepsilon}|\B|.\]
And if $\cN$ satisfies the second condition in Equation \eqref{eq:S-exchange}, then
\[ S\left(\widehat{\cN}\left(\frac{\Id_A}{|\A|}\right)\right) \leq S\left(\cN\left(\frac{\Id_A}{|\A|}\right)\right) +\varepsilon = \varepsilon,\]
which implies that
\[\log r_K\big(\widehat{\cN}\big) \geq \log|\A|-\varepsilon,\ \text{i.e.}\ r_K\big(\widehat{\cN}\big) \geq e^{-\varepsilon}|\A|.\]
Hence, the conclusion of this study is that, in Theorem \ref{th:entropies}, $r_K\big(\widehat{\cN}\big)\geq(1-\varepsilon)\max(|\A|,|\B|)$ is for sure necessary, in general, to have the entropy approximation \eqref{eq:approx-S}. It additionally tells us that there is a channel-dependent lower bound on $r_K\big(\widehat{\cN}\big)$ so that the latter holds (and hence even more so that the stronger notion of approximation in $(1{\rightarrow}1)$-norm holds), namely
\begin{equation} \label{eq:r_K-hat} \log r_K\big(\widehat{\cN}\big) \geq (1-\varepsilon) \max\left\{ \left|S(\varrho)-S(\cN(\varrho))\right| \st \varrho\in\cD(\A) \right\}. \end{equation}

Let us point out though that what we established here is optimality of our results only in the sense that there exist some quantum channels for which $\max(|\A|,|\B|)$ is necessary as approximating Kraus rank. It is however a more subtle question to find, for each given quantum channel $\cN$, what is the optimal approximating Kraus rank $\hat{r}_K(\cN)$. We leave this issue open, but in view of Equation \eqref{eq:r_K-hat}, a possible conjecture could be that 
\[ \log \hat{r}_K(\cN) \simeq \max \left\{ \left|S(\varrho)-S(\cN(\varrho))\right| \st \varrho\in\cD(\A) \right\}. \]
We use the occasion to also formulate an information-theoretic version of this question, namely: we wish to $\varepsilon$-approximate the channel $\cN^{\otimes n}$ (for concreteness, say in $(1{\rightarrow}1)$-norm) by one of minimum Kraus rank $\hat{r}_K(\cN,n,\varepsilon)$, and we would like to determine the value of
\[ R(\cN) := \underset{\varepsilon>0}{\sup}\, \underset{n\rightarrow+\infty}{\limsup}\, \frac{1}{n}\, \log \hat{r}_K(\cN,n,\varepsilon). \]
The latter quantity has the natural operational interpretation, as the minimum rate of qubits  needed in the environment, per channel realisation, to approximate many copies of the channel. One could thus hope to get information theoretic lower and upper bounds on it. The above reasoning shows, for instance, that
\[
  R(\cN) \geq \max \left\{ \left|S(\varrho)-S(\cN(\varrho))\right| \st \varrho\in\cD(\A) \right\}, 
\]
but we will not develop this notion further in the present paper.

\begin{proposition} \label{prop:F}
Let $\cN:\cL(\A)\rightarrow\cL(\B)$ be a CPTP map, and assume that the CP map $\widehat{\cN}:\cL(\A)\rightarrow\cL(\B)$ satisfies
\begin{equation} \label{eq:order'} \forall\ \varrho\in\cD(\A),\ (1-\varepsilon)\cN(\varrho) -\varepsilon\frac{\Id}{|\B|} \leq \widehat{\cN}(\varrho)\leq (1+\varepsilon)\cN(\varrho) +\varepsilon\frac{\Id}{|\B|},  \end{equation}
for some $0<\varepsilon<1/2$. Then, $\widehat{\cN}$ is close to $\cN$ in terms of output fidelities, in the sense that
\[ \forall\ \varrho\in\cD(\A),\forall\ \omega\in\cD(\B),\ \left| F\big(\widehat{\cN}(\varrho),\omega\big) - F\big(\cN(\varrho),\omega\big) \right| \leq \frac{3}{\sqrt{2}}\sqrt{\varepsilon}. \]
\end{proposition}

\begin{proof}
As noted in the proof of Proposition \ref{prop:S_p}, setting $\sigma=\cN(\varrho)$, $\widehat{\sigma}=\widehat{\cN}(\varrho)$ and $\tau=\Id/|\B|$, we can re-write Equation \eqref{eq:order'} as the two inequalities $\widehat{\sigma}\leq (1+\varepsilon)\sigma +\varepsilon\tau$ and $\sigma\leq (1+2\varepsilon)\widehat{\sigma} +2\varepsilon\tau$.
By operator monotonicity of $F(\cdot,\omega)$, and the fact that it is upper bounded by $1$, these imply the two estimates
\[ F(\widehat{\sigma},\omega) \leq \sqrt{1+\varepsilon}F(\sigma,\omega) +\sqrt{\varepsilon}F(\tau,\omega) \leq  F(\sigma,\omega) + \frac{\varepsilon}{2} + \sqrt{\varepsilon}, \]
\[ F(\sigma,\omega) \leq \sqrt{1+2\varepsilon}F(\widehat{\sigma},\omega) +\sqrt{2\varepsilon}F\big(\tau,\omega) \leq  F(\widehat{\sigma},\omega) + \varepsilon + \sqrt{2\varepsilon}. \]
Finally, just observing that $\varepsilon\leq \sqrt{\varepsilon/2}$ for $0<\varepsilon<1/2$, the conclusion of Proposition \ref{prop:F} directly follows.
\end{proof}

\begin{theorem} \label{th:fidelitiesies}
Fix $0<\varepsilon<1$ and let $\cN:\cL(\A)\rightarrow\cL(\B)$ be a CPTP map with Kraus rank $|\rE|\geq |\A|,|\B|$. Then, there exists a CP map $\widehat{\cN}:\cL(\A)\rightarrow\cL(\B)$ with Kraus rank at most $C\max(|\A|,|\B|)\log(|\rE|/\varepsilon)/\varepsilon^4$ (where $C>0$ is a universal constant) and such that
\[ \forall\ \varrho\in\cD(\A),\forall\ \omega\in\cD(\B),\ \left| F\big(\widehat{\cN}(\varrho),\omega\big) - F\big(\cN(\varrho),\omega\big) \right| \leq \varepsilon. \]
\end{theorem}

\begin{proof}
This is a direct consequence of Theorem \ref{th:main}, combined with Proposition \ref{prop:F}.
\end{proof}

\subsection{Destruction of correlations with few resources and data hiding.} \hfill\smallskip

It was observed in \cite{HLSW}, Section 3, that an $\varepsilon$-randomizing channel (i.e.~a channel which is an $\varepsilon$-approximation of the fully randomizing channel) approximately destroys the correlations between the system it acts on and any system the latter might be coupled to, in the following two senses: First of all, a state which is initially just classically correlated becomes almost uncorrelated (or in other words any separable state is sent close to a product state, in $1$-norm distance). And secondly, whatever the initial state, the correlations present in it become almost invisible to local observers (or in other words any state is sent to close to a product state, in one-way-LOCC-norm). Hence, having an $\varepsilon$-randomizing channel with few Kraus operators can be seen as having an efficient way to decouple a system of interest from its environment. Thanks to Theorem \ref{th:main}, we can generalize these results to Theorem \ref{th:correlations} below.

\begin{theorem} \label{th:correlations}
Let $\A,\B,\mathrm{C}$ be Hilbert spaces, and assume that $d=\max(|\A|,|\B|)<+\infty$. For any $0<\varepsilon<1$ and $\sigma_B^*\in\mathcal{D}(\B)$, there exists a CPTP map $\widehat{\mathcal{N}}:\mathcal{L}(\A)\rightarrow\mathcal{L}(\B)$ with Kraus rank at most $Cd\log(d/\varepsilon)/\varepsilon^2$ (where $C>0$ is a universal constant) and such that
\begin{equation} \label{eq:correlations1} \forall\ \varrho_{AC}\in\mathcal{S}(\A:\mathrm{C}),\ \left\|\widehat{\mathcal{N}}\otimes\mathcal{I}(\varrho_{AC})-\sigma_B^*\otimes\varrho_C\right\|_1 \leq \varepsilon, \end{equation}
\begin{equation} \label{eq:correlations2} \forall\ \varrho_{AC}\in\mathcal{D}(\A\otimes\mathrm{C}),\ \left\|\widehat{\mathcal{N}}\otimes\mathcal{I}(\varrho_{AC})-\sigma_B^*\otimes\varrho_C\right\|_{\mathbf{LOCC^{\rightarrow}(\B:\mathrm{C})}} \leq \varepsilon. \end{equation}
\end{theorem}

\begin{proof}
Define the completely forgetful CPTP map $\mathcal{N}:X_A\in\mathcal{L}(\A)\mapsto(\tr X_A)\sigma_B^*\in\mathcal{L}(\B)$ (i.e.~$\mathcal{N}$ sends every input state to the output state $\sigma_B^*$). By Theorem \ref{th:main}, there exists a CPTP map $\widehat{\mathcal{N}}:\mathcal{L}(\A)\rightarrow\mathcal{L}(\B)$ with Kraus rank at most $Cd\log(d/\varepsilon)/\varepsilon^2$ such that
\[ \forall\ \varrho_A\in\mathcal{D}(\A),\ \left\|\widehat{\mathcal{N}}(\varrho_A)-\mathcal{N}(\varrho_A)\right\|_1\leq\varepsilon\ \text{i.e.}\ \left\|\widehat{\mathcal{N}}(\varrho_A)-\sigma_B^*\right\|_1\leq\varepsilon. \]
Now, following the exact same route as in the proofs of Lemmas III.1 and III.2 in \cite{HLSW}, we get that this implies precisely Equations \eqref{eq:correlations1} and \eqref{eq:correlations2}, respectively. We will therefore only briefly recall the arguments here.

Concerning Equation \eqref{eq:correlations1}, let $\varrho_{AC}\in\mathcal{S}(\A:\mathrm{C})$, i.e.~$\varrho_{AC}=\sum_xp_x\varrho_A^{(x)}\otimes\varrho_C^{(x)}$. Then,
\begin{align*}
    \left\|\widehat{\mathcal{N}}\otimes\mathcal{I}(\varrho_{AC})-\sigma_B^*\otimes\varrho_C\right\|_1 & = \left\|\sum_x p_x \left(\widehat{\mathcal{N}}\left(\varrho_{A}^{(x)}\right) -\sigma_B^*\right)\otimes\varrho_C^{(x)} \right\|_1 \\
    & \leq \sum_x p_x \left\|\widehat{\mathcal{N}}\left(\varrho_A^{(x)}\right)-\sigma_B^*\right\|_1 \\
    & \leq \varepsilon,
\end{align*}  
where the last inequality is because, by assumption, for each $x$, $\big\|\widehat{\mathcal{N}}\big(\varrho_A^{(x)}\big)-\sigma_B^*\big\|_1 \leq \varepsilon$, and $\sum_x p_x=1$.

As for Equation \eqref{eq:correlations2}, let $\mathrm{M}=\big(M_B^{x}\otimes M_C^{(x)}\big)_x\in\mathbf{LOCC^{\rightarrow}(\B:\mathrm{C})}$, i.e.~for each $x$, $0\leq M_B^{(x)},M_C^{(x)}\leq \Id$, and $\sum_x M_B^{(x)}=\Id$. Then, for any $\varrho_{AC}\in\mathcal{D}(\A\otimes\mathrm{C})$,
\begin{align*} & \left\| \widehat{\mathcal{N}}\otimes\mathcal{I}(\varrho_{AC})-\mathcal{N}\otimes\mathcal{I}(\varrho_{AC})\right\|_{\mathrm{M}} \\
& \qquad =\sum_x \left| \tr\left[ M_B^{(x)}\otimes M_C^{(x)} \left( \widehat{\mathcal{N}}\otimes\mathcal{I}(\varrho_{AC})- \mathcal{N}\otimes\mathcal{I}(\varrho_{AC}) \right) \right] \right| \\
& \qquad =\sum_x \left| \tr\left[ \left(\widehat{\mathcal{N}}^*\left(M_B^{(x)}\right)- \mathcal{N}^*\left(M_B^{(x)}\right)\right)\otimes M_C^{(x)} \varrho_{AC}\right] \right| \\
& \qquad \leq \sum_x \left\| \widehat{\mathcal{N}}^*\left(M_B^{(x)}\right)- \mathcal{N}^*\left(M_B^{(x)}\right) \right\|_{\infty} \\
& \qquad \leq \varepsilon,
\end{align*}
where the next-to-last inequality is because, $\|\varrho_{AC}\|_1\leq 1$ and for each $x$, $\big\|M_C^{(x)}\big\|_{\infty}\leq 1$, while the last inequality is because, by assumption, for each $x$, $\big\| \widehat{\mathcal{N}}^*\big(M_B^{(x)}\big)- \mathcal{N}^*\big(M_B^{(x)}\big) \big\|_{\infty} \leq \varepsilon \tr M_B^{(x)}/|\B|$, and $\sum_x \tr M_B^{(x)}=|\B|$.
\end{proof}

\begin{remark}
Note that the completely forgetful CPTP map
\[ \mathcal{N}:X_A\in\mathcal{L}(\A)\mapsto(\tr X_A)\sigma_B^*\in\mathcal{L}(\B) \] 
has Kraus rank equal to $|\A|\times\mathrm{rank}(\sigma_B^*)$. Hence, the content of Theorem \ref{th:correlations} is interesting only for states $\sigma_B^*$ of sufficiently high rank, namely $\mathrm{rank}(\sigma_B^*)\geq C d\log d/|\A|$.
\end{remark}

In a similar vein, Theorem \ref{th:main'} implies that any bipartite state which is sufficiently mixed can be used for data hiding with a bipartite state of much lower rank. In brief, two bipartite states are called data hiding if there exists a global measurement that allows to distinguish them very well, while they remain poorly distinguishable by all local measurements and classical communication. Most known examples of good data hiding states (such as Werner states or random states) are of high rank. Theorem \ref{th:data-hiding} below shows how to construct, for any `truly' high rank state, an associated low rank state such that the pair is data hiding.

\begin{theorem} \label{th:data-hiding}
Fix $0<\varepsilon<1$ and let $\tau\in\mathcal{D}(\A\otimes\A)$ be such that $\|\tau\|_{\infty} \leq C/|\A|^2$ (where $C>0$ is a universal constant). Then, there exists a state $\hat{\tau}\in\mathcal{D}(\A\otimes\A)$, with rank at most $C'|\A|/\varepsilon^2$ (where $C'>0$ is a universal constant),  satisfying 
\begin{equation} \label{eq:data-hiding}  
\left\| \tau-\hat{\tau} \right\|_{\mathbf{LOCC^{\rightarrow}(\A:\A)}} \leq C\varepsilon \ \text{and}\ \left\| \tau-\hat{\tau} \right\|_1 \geq 2\left( 1 - \frac{C'C}{\varepsilon^2|\A|} \right).
\end{equation}
\end{theorem}

\begin{proof}
Let $\mathcal{N}:\mathcal{L}(\A) \rightarrow \mathcal{L}(\A)$ be the quantum channel whose Choi-Jamio\l{}kowski state is $\tau$. It is easy to check that, since $\|\tau\|_{\infty} \leq C/|\A|^2$, $\mathcal{N}$ is such that
\[ \forall\ \varrho\in\mathcal{D}(\A),\ \|\mathcal{N}(\varrho)\|_{\infty} \leq \frac{C}{|\A|}. \]
Indeed, by definition of $\tau$, for any $\varrho\in\mathcal{D}(\A)$,
\begin{align*}
    \|\mathcal{N}(\varrho)\|_{\infty} & = \max_{\|X\|_1\leq 1} \tr(\mathcal{N}(\varrho)X) \\
    & = |\A| \times \max_{\|X\|_1\leq 1} \tr(\tau\,X\otimes\varrho^T) \\
    & \leq |\A| \times \max_{\|Y\|_1\leq 1} \tr(\tau Y) \\
    & = |\A| \times \|\tau\|_{\infty} \\
    & \leq \frac{C}{|\A|}.
\end{align*} 
Hence, by Theorem \ref{th:main'}, there exists $\widehat{\mathcal{N}}$ with Kraus rank at most $C'|\A|/\varepsilon^2$ such that
\[ \forall\ \varrho\in\mathcal{D}(\A),\ \| \mathcal{N}(\varrho) - \widehat{\mathcal{N}}(\varrho) \|_1 \leq C\varepsilon. \]
Let $\hat{\tau}\in\mathcal{D}(\A\otimes\A)$ be the Choi-Jamio\l{}kowski state associated to $\widehat{\mathcal{N}}$. It has rank equal to the Kraus rank of $\widehat{\mathcal{N}}$, i.e.~at most $C'|\A|/\varepsilon^2$. By the same reasoning as in the proof of Equation \eqref{eq:correlations2} in Theorem \ref{th:correlations}, we have, denoting by $\psi$ the maximally entangled state on $\A\otimes\A$,
\[ \left\| \tau-\hat{\tau} \right\|_{\mathbf{LOCC^{\rightarrow}(\A:\A)}} = \left\| \mathcal{N}\otimes\mathcal{I}(\psi) - \widehat{\mathcal{N}}\otimes\mathcal{I}(\psi) \right\|_{\mathbf{LOCC^{\rightarrow}(\A:\A)}} \leq C\varepsilon. \]
This proves the first inequality in Equation \eqref{eq:data-hiding}. As for the second one, it follows from the fact that $\tau$ has largest eigenvalue at most $C/|\A|^2$, while $\hat{\tau}$ has rank at most $C'|\A|/\varepsilon^2$. Indeed, we therefore have
\[ \left\| \tau-\hat{\tau} \right\|_1 \geq \frac{C'|\A|}{\varepsilon^2}\left( \frac{\varepsilon^2}{C'|\A|} - \frac{C}{|\A|^2} \right) + 1 - \frac{C'C}{\varepsilon^2|\A|} = 2 \left( 1 - \frac{C'C}{\varepsilon^2|\A|} \right), \]
which is exactly the announced result.
\end{proof}

Applying Theorem \ref{th:data-hiding} with $\varepsilon=1/|\A|^{\alpha}$, for some $0<\alpha<1/2$, shows that, for any state $\tau$ on $\A\otimes\A$ such that $\|\tau\|_{\infty} \leq C/|\A|^2$, we can construct a state $\hat{\tau}$ on $\A\otimes\A$ with rank at most $C'|\A|^{1+2\alpha}$ satisfying 
\[ \left\| \tau-\hat{\tau} \right\|_{\mathbf{LOCC^{\rightarrow}(\A:\A)}} \leq \frac{C}{|\A|^{\alpha}} \ \text{and}\ \left\| \tau-\hat{\tau} \right\|_1 \geq 2\left( 1 - \frac{C'}{|\A|^{1-2\alpha}} \right). \]

\subsection{Werner channels.} \hfill\smallskip

An interesting case to which Theorem \ref{th:main'} applies is that of the so-called Werner channels. These are defined as the family of CPTP maps
\[ \mathcal{W}_{\lambda}:X\in\mathcal{L}(\A)\mapsto\frac{1}{|\A|+2\lambda-1}\left[(\tr X)\Id +(2\lambda-1)X^T\right]\in\mathcal{L}(\A),\ 0\leq\lambda\leq 1. \]
Denoting by $\varsigma$ and $\alpha$ the symmetric and anti-symmetric states on $\A\otimes\A$, it is easy to check that, for each $0\leq\lambda\leq 1$, the Choi-Jamio\l{}kowski state $\tau(\mathcal{W}_{\lambda})$ associated to $\mathcal{W}_{\lambda}$ is nothing else than the Werner state $\rho_{\lambda}=\lambda\varsigma+(1-\lambda)\alpha$. Hence, $\mathcal{W}_{\lambda}$ has Kraus rank $|\A|^2$ whenever $0<\lambda<1$, and $|\A|(|\A|+1)/2$, resp. $|\A|(|\A|-1)/2$, when $\lambda=1$, resp. $\lambda=0$, i.e.~in any case full or almost full Kraus rank. These channels are thus typically of the kind that we would like to compress into more economical ones. What is more, they have the property of having only very mixed output states. Indeed,
\[ \max_{\varrho\in\mathcal{D}(\A)}\left\|\mathcal{W}_{\lambda}(\varrho)\right\|_{\infty} = \begin{cases} 2\lambda/(|\A|+2\lambda-1)\ \text{if}\ \lambda\geq 1/2\\ 1/(|\A|+2\lambda-1)\ \text{if}\ \lambda< 1/2 \end{cases} \leq \frac{2}{|\A|}. \]
So by Theorem \ref{th:main'}, we get that, for each $0\leq\lambda\leq 1$, given $0<\varepsilon<1$, there exists a CP map $\widehat{\mathcal{W}}_{\lambda}:\mathcal{L}(\A)\rightarrow\mathcal{L}(\A)$ with Kraus rank at most $C|\A|/\varepsilon^2$ (where $C>0$ is a universal constant) such that
\[ \forall\ \varrho\in\mathcal{D}(\A),\ \left\| \widehat{\mathcal{W}}_{\lambda}(\varrho)-\mathcal{W}_{\lambda}(\varrho)\right\|_{\infty} \leq \frac{\varepsilon}{|\A|}. \]
In words, this means that the Werner CPTP maps can be $(\varepsilon/|\A|)$-approximated in $(1{\rightarrow}\infty)$-norm distance (hence in particular $\varepsilon$-approximated in $(1{\rightarrow}1)$-norm distance) by CP maps having Kraus rank $C|\A|/\varepsilon^2\ll |\A|^2$.

\begin{remark}
Note that, as a special case of this approximation result for Werner channels, we recover the well-known approximation result for the fully randomizing channel $\mathcal{R}$, recalled in Section \ref{sec:background}. Indeed, $\mathcal{W}_{1/2}$ is nothing else than $\mathcal{R}$.
\end{remark}

\section{Discussion}
\label{sec:discussion}

We have generalized in several senses the result established in \cite{HLSW} and \cite{Aubrun}. First, we have shown that it holds for all quantum channels and not only for the fully randomizing one: any CPTP map from $\mathcal{L}(\mathrm{A})$ to $\mathcal{L}(\mathrm{B})$ can be $\varepsilon$-approximated in $(1{\rightarrow}1)$-norm distance by a CPTP map with Kraus rank of order $d\log(d/\varepsilon)/\varepsilon^2$, where $d=\max(|\A|,|\B|)$. Second, we have established that a stronger notion of approximation can actually be proven, namely an $\varepsilon$-ordering of the two CP maps, which allows to derive approximation results in terms of various output quantities (that are tighter than those induced by the rougher norm distance closeness). In the case where the channel under consideration is, as the fully randomizing channel, very noisy (meaning that all output states are very mixed), the extra $\log(d/\varepsilon)$ factor in our result can be removed. However, we do not know if this is true in general. On a related note, our study of optimality shows that there exist channels which cannot be compressed below order $d$ Kraus operators (even to achieve the weakest notions of approximation). But what about channel-dependent upper and lower bounds? For a given channel, would there be a more clever construction than ours (i.e.~a non-universal one) that would enable its compression to a number of Kraus operators whose $\log$ would be, for instance, of the order of its maximum input-output entropy difference?

\smallskip
Furthermore, full or partial derandomization of our construction would be desirable in practice. Here again the main difficulty is that most of the techniques which apply to very noisy channels may fail in general. Let us specify a bit what we mean. In \cite{Aubrun}, two approximation schemes were proposed for the fully randomizing channel $\mathcal{R}:\cL(\C^d)\rightarrow\cL(\C^d)$. They consisted in taking as Kraus operators $\{U_i/\sqrt{n},\ 1\leq i\leq n\}$ with $U_1,\ldots,U_n$ sampled either from the Haar measure on $\cU(\C^d)$ or from any other isotropic (aka unitary $1$-design) measure on $\cU(\C^d)$. It was then shown that, in order to approximate $\mathcal{R}$ up to error $\varepsilon/d$ in $(1{\rightarrow}\infty)$-norm, $n$ of order $d/\varepsilon^2$ was enough in the Haar-distributed case and $n$ of order $d\log^6d/\varepsilon^2$ was enough in the, more general, isotropically-distributed case. The advantage of the second result compared to the first one is that there exist isotropic measures which are much simpler than the Haar measure on $\cU(\C^d)$, in particular discrete ones (e.g.~the uniform measure over any unitary orthogonal basis of $\cL(\C^d)$). Hence, from a practical point of view, generating such a measure is arguably more realistic than generating the Haar measure (the reader is e.g.~referred to \cite{BHH} for a more precise formulation of the claim that implementing a Haar distributed unitary is hard and an extensive discussion on how to approximate such a unitary by a more easily implementable one). Now, if $\cN:\cL(\C^d)\rightarrow\cL(\C^d)$ is a channel, with environment $\C^s$, such that $\sup_{\rho\in\cD(\C^d)}\|\cN(\rho)\|_{\infty} \leq C/d$, then arguments of the same type apply to our construction: to approximate $\mathcal{N}$ up to error $\varepsilon/d$ in $(1{\rightarrow}\infty)$-norm by sampling unit vectors in $\C^s$, order $d/\varepsilon^2$ of them is enough if they are Haar-distributed (which is the content of Theorem \ref{th:main}) and order $d\log^6d/\varepsilon^2$ of them is enough if they are only assumed to be isotropically-distributed. Here as well, the gain in terms of needed amount of randomness is obvious: there exist isotropic measures which are much simpler to sample from than the Haar-measure on $S_{\C^s}$ (e.g.~the uniform measure on any orthonormal basis of $\C^s$). Unfortunately, this whole reasoning (based on Dudley's upper bounding of Bernoulli averages by covering number integrals and on a sharp entropy estimate for the suprema of empirical processes in Banach spaces) fails completely for channels that have some of their outputs which are too pure.

\smallskip
Finally, one could ask whether other parameters than the Kraus rank would be relevant ones to try and minimize. As we have already explained, the question we investigate can be seen as a data compression problem: how can we approximate a given protocol while reducing the amount of resources needed? Taking the Kraus rank, an integer-valued non-smooth quantity, as figure of merit in this channel compression task is somehow putting ourselves in a one-shot scenario. But one could wonder how to define an asymptotic version of this question, with a corresponding regularized version of the Kraus rank.

\section*{Acknowledgements}
This work owes a lot to Guillaume Aubrun. He is actually the one who first raised the question of compressing a general quantum channel in terms of its Kraus operators, following the papers \cite{HLSW} and \cite{Aubrun} which were dealing with the case of the fully randomizing channel. He then followed all the steps of the project, including the most off-putting ones, such as proof-reading technical points of the paper. 
We would also like to thank an anonymous referee for their incredibly constructive comments, which have helped us improve the paper considerably. In particular, the alternative way of viewing our construction, in terms of Kraus operators rather than Stinespring dilation, that we discuss in Remark \ref{rem:Kraus}, is due to them.

This research was supported by the European Research Council (grants no.~337603 and no.~648913), the Spanish MINECO (projects FIS2013-40627-P,
FIS2016-86681-P and PID2019-107609GB-I00) with the support of FEDER funds, the Generalitat de Catalunya (CIRIT projects 
2014-SGR-966 and 2017-SGR-1127), the John Templeton Foundation (grant no.~48322), the French ANR (projects Stoq 14-CE25-0033 and ANR-11-LABX-0040).

\end{document}